\documentclass[prl,twocolumn,superscriptaddress]{revtex4-1}
\usepackage{amsthm,amsmath,amsfonts,amssymb}
\usepackage{bm,color,appendix,graphicx,caption}
\usepackage[colorlinks,bookmarks=true,citecolor=blue,linkcolor=blue,urlcolor=blue]{hyperref}
\newtheorem{theorem}{Theorem}[section]
\newtheorem{corollary}{Corollary}[theorem]

\newcommand{\beqn}{\begin{eqnarray}}
\newcommand{\eeqn}{\end{eqnarray}}

\newcommand{\mG}{\mathcal{G}}
\newcommand{\mC}{\mathcal{C}}
\newcommand{\tw}{\tilde{w}}

\begin{document}
\title{Exact Landau Level Description of Geometry and Interaction in a Flatband}

\author{Jie Wang}
\email{jiewang@flatironinstitute.org}
\affiliation{Center for Computational Quantum Physics, Flatiron Institute, 162 5th Avenue, New York, New York 10010, USA}
\author{Jennifer Cano}
\affiliation{Center for Computational Quantum Physics, Flatiron Institute, 162 5th Avenue, New York, NY 10010, USA}
\affiliation{Department of Physics and Astronomy, Stony Brook University, Stony Brook, New York 11974, USA}
\author{Andrew J. Millis}
\affiliation{Center for Computational Quantum Physics, Flatiron Institute, 162 5th Avenue, New York, NY 10010, USA}
\affiliation{Department of Physics, Columbia University, 538 W 120th Street, New York, New York 10027, USA}
\author{Zhao Liu}
\email{zhaol@zju.edu.cn}
\affiliation{Zhejiang Institute of Modern Physics, Zhejiang University, Hangzhou 310027, China}
\author{Bo Yang}
\email{yang.bo@ntu.edu.sg}
\affiliation{Division of Physics and Applied Physics, Nanyang Technological University, 637371, Singapore}
\affiliation{Institute of High Performance Computing, A$^*$STAR, 138632, Singapore}

\begin{abstract}
Flatbands appear in many condensed matter systems, including the two-dimensional electron gas in a high magnetic field, correlated materials, and moir\'e heterostructures. They are characterized by intrinsic geometric properties such as the Berry curvature and Fubini-Study metric. The influence of the band geometry on electron-electron interaction is difficult to understand analytically because the geometry is in general nonuniform in momentum space. In this work, we study the topological flatband of Chern number $~\mC=1$ with a momentum-dependent but positive definite Berry curvature that fluctuates in sync with Fubini-Study metric. We derive an exact correspondence between such ideal flatbands and Landau levels and show that the band geometry fluctuation gives rise to a new type of interaction in the corresponding Landau levels that depends on the center of mass of two particles. We characterize such interactions by generalizing the usual Haldane pseudopotentials. This mapping gives exact zero-energy ground states for short-ranged repulsive generalized pseudopotentials in flatbands, in analogy to fractional quantum Hall systems. Driving the center-of-mass interactions beyond the repulsive regime leads to a dramatic reconstruction of the ground states towards gapless phases. The generalized pseudopotential could be a useful basis for future numerical studies.
\end{abstract}

\maketitle

The one-electron states in periodic solids are characterized both by their dispersion (variation of energy with crystal momentum) and by their band geometry, defined by the variation of the electronic wavefunction with crystal momentum. In a single-band system, the band geometry is defined by the quantum geometric tensor:
\begin{equation}
\mathcal{Q}^{ab}_{\bm k} = \langle D^a_{\bm k}u_{\bm k}|D^b_{\bm k}u_{\bm k}\rangle = g^{ab}_{\bm k} + \frac{i}{2}\epsilon^{ab}\Omega_{\bm k},\label{quantumgeotensor}
\end{equation}
where $u_{\bm k}(\bm r)=\langle\bm r|u_{\bm k}\rangle$ is the periodic part of the Bloch wavefunction $\psi_{\bm k}(\bm r)$, $D^a_{\bm k}$ is the covariant derivative operator that adiabatically transports the wavefunction along the spatial direction $a=x,y$, and $\epsilon^{ab}$ is the antisymmetric tensor. Here the Berry curvature $\Omega_{\bm k}$, and the Fubini-Study metric (FSM) $g^{ab}_{\bm k}$, are respectively the imaginary and real part of the quantum geometric tensor \footnote{We use the convention that Berry connection and Berry curvature are $\bm A^a_{\bm k}=-i\protect\langle u_{\bm k}|\partial^a_{\bm k}u_{\bm k}\protect\rangle$ and $\Omega_{\bm k}=\epsilon_{ab}\partial^a_{\bm k}\bm A^b_{\bm k}$ respectively, which differs from the usual convention by a minus sign, but gives us $\bm k-$space holomorphic wavefunctions with positive Berry curvature. This sign convention was used, for instance, in Ref.~(\onlinecite{haldaneanomaloushall}).}.

The interplay of the band geometry and dispersion has been elucidated on the single-particle level \cite{vanderbilt_2018,Di_review,Peotta:2015aa}, where it leads to many interesting phenomena including the anomalous Hall effect \cite{RMP_AQH,haldaneanomaloushall}. Recent experimental and theoretical interest on moir\'e materials has centered on the ``flatband" situation \cite{Bistritzer12233,Santos,Cao:2018aa,Cao:2018ab,Sharpe605,Serlin900}, where the electron dispersion is small relative to interaction scales and the physics is controlled by electron-electron interactions. A growing body of evidence indicates that in flatband situations the band geometry plays a crucial role in determining the electron-electron interaction physics. For example, in the canonical lowest Landau level (LLL) problem of electrons with a continuous two-dimensional translation invariance in a uniform perpendicular magnetic field, both $g^{ab}_{\bm k}$ and $\Omega_{\bm k}$ are $\bm k-$independent. This $\bm k-$independence enables detailed analytical understanding of the physics even in the presence of strong electron-electron interactions \cite{Haldane_hierarchy}. However, generically in periodic lattice systems the band geometry is highly nonuniform in momentum space, and while the interplay between the band geometry and interactions has been numerically studied \cite{PhysRevX.1.021014,zhao_review} analytical understanding has been limited \cite{PARAMESWARAN2013816,PhysRevB.90.165139,Jackson:2015aa,Martin_PositionMomentumDuality,Yangle_Modelwf,PhysRevB.88.035101,PhysRevB.90.085103,PhysRevB.96.165150,PhysRevX.5.041003,scottjiehaldane,Jie_Dirac}. 

In this Letter we take a step towards understanding the relation between the band geometry and interaction physics. Our work is inspired by the chiral model of  twisted bilayer graphene (cTBG) which at certain ``magic'' twist angles realizes exactly dispersionless bands \cite{Grisha_TBG}. The chiral model is understood as a kind of fixed point Hamiltonian \cite{Oscar_hiddensym} capturing the interacting physics of twisted bilayer graphene \cite{Zaletel_PRX20,lian2020tbg,bernevig2020tbg_5} and has the special property \cite{PhysRevB.90.165139,Jackson:2015aa,Martin_PositionMomentumDuality,Grisha_TBG2} that while the band geometry is nonuniform, the FSM is related to the Berry curvature in the following way:
\begin{equation}
g^{ab}_{\bm k} = \frac{1}{2}\omega^{ab}\Omega_{\bm k},\label{tracecond}
\end{equation}
where $\omega^{ab}$ is a constant determinant one positive symmetric matrix.

Following Ref.~(\onlinecite{Martin_PositionMomentumDuality}), we define \emph{ideal flatbands} as dispersionless Bloch bands with (i) a positive definite Berry curvature that (ii) fluctuates in sync with the FSM as in Eq.~(\ref{tracecond}). We show that the ideal flatband assumptions (i) and (ii) fix the forms of single-particle wavefunctions in a topological flatband with Chern number $\mC=1$ (as occurs in cTBG), establishing an exact correspondence between an ideal flatband and the LLL. Using this correspondence we show that the electron-electron interaction in an ideal flatband with spatially fluctuating band geometry can be exactly mapped to a center-of-mass (COM) dependent interaction in the LLL, which can be systematically characterized by the generalized COM pseudopotentials derived here. We show that the resulting interacting Hamiltonian possesses exact zero modes, corresponding to the previously discussed generalization of the Laughlin fractional quantum Hall (FQH) states \cite{Grisha_TBG2}; however, depending on the values of the COM interaction parameters, charge density wave states of lower energy may exist. In the last section of the Supplementary Material (SM), we derive further implications for superconductivity and the composite Fermi liquid phase in TBG flatbands.

\emph{Wavefunctions of $~\mC=1$ ideal flatbands.---}
% We first discuss the special properties of ideal flatbands, then derive the one-electron wavefunction and show its generality. We consider a flatband satisfying Eq.~(\ref{tracecond}) with positive-definite Berry curvature.
% \footnote{Clarifications on ideal flatband conditions (i) positive definite Berry curvature; (ii) Eq.~(\ref{tracecond}). The positive semi-definite property of the quantum geometry tensor $\mathcal{Q}^{ab}_{\bm k}$ leads to a ``determinant bound'': $\det g_{\bm k}\geq|\Omega_{\bm k}|^2/4$. The saturation of the det-bound implies the existence of $\bm k-$space local complex structure $w_{\bm k,a}$ provided $\Omega_{\bm k}\neq0$, where $w_{\bm k,a}$ is also the null-vector of $\mathcal{Q}^{ab}_{\bm k}$. Note that any 2D two-band model saturates the det-bound but $\Omega_{\bm k}$ must vanish identically somewhere in BZ due to topological reasons \cite{kahlerband1,kahlerband2}. Therefore, only by demanding both (i) and saturating of the det-bound can give raise to a well defined global holomorphic coordinate system in the $\bm k-$space \cite{kahlerband1,kahlerband2}. Mapping to LLs is however still nontrivial because of the $\bm k-$dependence in $\omega_{\bm k,a}$. To simplify, we assume a stronger condition (ii) which means $\omega_{\bm k,a}$ is $\bm k-$independent, and gives raise to the wavefunction Eq.~(\ref{martin_wf}). See SM for more details.}
Locally, such a flatband mimics a LL in $\bm k-$space: the quantum geometric tensor at every $\bm k$ point has a constant null vector $\mathcal{Q}^{ab}_{\bm k}\omega_b=0$, which also determines $\omega^{ab}=\omega^a\omega^{b*} + \omega^{a*}\omega^b$ in Eq.~(\ref{tracecond}). This uniform null vector defines the $\bm k-$space complex structure \cite{kahlerband1,kahlerband2} and gives the Bloch wavefunction a universal form \cite{Martin_PositionMomentumDuality}:
\begin{equation}
\psi_{\bm k}(\bm r) \sim \tilde{u}_{k}(\bm r)\exp(i\bm k\cdot\bm r),\label{martin_wf}
\end{equation}
where the bolded $\bm k$ gives the momentum vector and unbolded $k\equiv\omega^a\bm k_a$ is a complex number. The cell-periodic function $\tilde{u}_k$ is holomorphic in $k$ up to a normalization factor.

We define the $\bm k-$space boundary condition for the periodic part of the Bloch wavefunction:
\begin{equation}
\tilde{u}_{k+b}(\bm r) = e^{i\phi_{k,b}}e^{-i\bm b\cdot\bm r}\tilde{u}_{k}(\bm r).
\end{equation}

The complex phase $\phi_{k,b}$ must be holomorphic in $k$ because both $\tilde{u}_{k+b}$ and $\tilde{u}_k$ are. A nonzero Chern number requires that $\tilde{u}_k$, as a function of $k$, must have discontinuities in the Brillouin zone (BZ). Such discontinuities show up at the BZ boundary as non-zero $\phi_{k,b}$, in the bulk as wavefunction singularities, or both \cite{Thouless_1984}. For a $\mC=1$ ideal flatband, it is necessary to have non-zero $\phi_{k,b}$: in contrast, Ref.~(\onlinecite{Martin_PositionMomentumDuality}) assumed $\phi_{k,b}=0$ so discussions were limited to wavefunctions of $\mC \geq 2$.

The boundary condition $\phi_{k,b}$ plays a crucial role in determining the wavefunction of the ideal band. Following Cauchy's argument principle, the BZ boundary integral $\frac{1}{2\pi i}\oint dk~\partial_k \ln\tilde{u}_{k}(\bm r)$ is an integer. We show in the SM that this integer is equal to the Chern number and can be written as \cite{Note3}:
\begin{equation}
\mC=-\frac{1}{2\pi}\left(\phi_{k_0+b_1,b_2}-\phi_{k_0,b_2}+\phi_{k_0,b_1}-\phi_{k_0+b_2,b_1}\right),\label{boundaryintegral}
\end{equation}
where, as illustrated in Fig.~\ref{illustration} (a), $b_{1,2}$ are primitive reciprocal lattice vectors and $k_0$ is the BZ origin. Insensitivity of the Chern number to the choice of $k_0$, combined with Eq.~(\ref{boundaryintegral}), forces $\phi_{k,b}$ to be a linear function of $k$. Since $\tilde{u}_k$ is holomorphic in $k$, it is uniquely determined by the boundary condition $\phi_{k,b}$, giving the bulk wavefunction:
\begin{equation}
\psi_{\bm k}(\bm r) = \mathcal{N}_{\bm k}\mathcal{B}(\bm r)\Phi_{\bm k}(\bm r),\label{modelwavefunction}
\end{equation}
where $\mathcal{N}_{\bm k}$, $\mathcal{B}(\bm r)$ and $\Phi_{\bm k}$ are the normalization factor, a $\bm k-$independent quasiperiodic function and the LLL wavefunction, respectively. Expressed in the symmetric gauge, $\Phi_{\bm k}(\bm r) = \sigma(z+ik) \exp\left(ik^*z\right) \exp\left(-\frac{1}{2}|z|^2-\frac{1}{2}|k|^2\right)$ where $\sigma(z)$ is the modified Weierstrass sigma function \cite{haldanemodularinv,Jie_MonteCarlo,FERRARI_Sigma,FERRARI_QHwannier} and $z\equiv\omega_a\bm r^a$ \footnote{Compared to the commonly used Jacobi $\theta$ function representation (Landau gauge), the $\sigma$ function representation (symmetric gauge) has the advantage of labeling LLL states by two quantum numbers $\bm k=(k_1, k_2)$ in the same way of labeling Bloch states in solids. Different representations are related by gauge transformations discussed in SM.}. Generalizing to negative definite Berry curvature is straightforward. We leave more detailed discussions of the ideal flatband conditions, holomorphic wavefunction Eq.~(\ref{martin_wf}) and the uniqueness of our $\mC=1$ model wavefunction Eq.~(\ref{modelwavefunction}) to the SM \footnote{SM includes detailed discussions on: model wavefunctions, ideal conditions, interacting models, pseudopotentials and numerical details.}.

\begin{figure}
    \centering
    \includegraphics[width=0.5\textwidth]{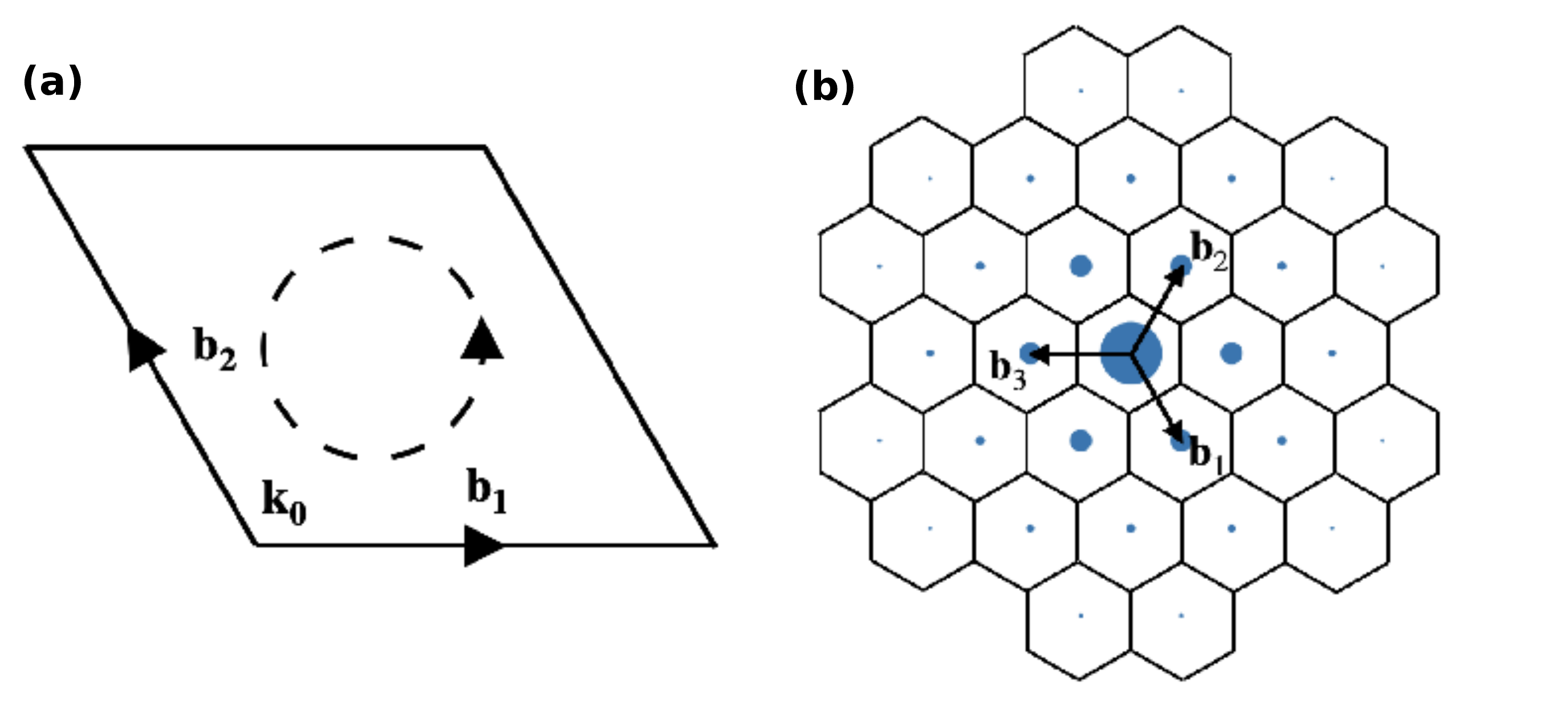}
    \caption{(a) Geometry of the Brillouin zone: $\bm k_0$ is the origin, $\bm b_{1,2}$ are primitive reciprocal lattice vectors, and the dashed circle sketches the orientation of the Brillouin zone boundary integral used in Eq.~(\ref{boundaryintegral}). (b) Plot of reciprocal space of cTBG at the first magic angle with reciprocal lattice vectors used in the main text indicated, and the values of the Fourier modes $w_{\bm b}$ Eq.~(\ref{FourierMode}) indicated by the size of solid dots. We find the first two modes $w_{\bm 0}=1$, $w_{\pm\bm b_{1,2,3}}=0.243$ dominate. The $w_{\bm b}$ determine the single-particle band geometry through Eq.~(\ref{expressionOmega}) and Eq.~(\ref{BCNK}), and the interaction model Eq.~(\ref{effectiveinteractionQH}).}\label{illustration}
\end{figure}

\emph{Band geometry of ideal flatbands.---}
We now explicitly compute the band geometry of an ideal flatband using the model wavefunction Eq.~(\ref{modelwavefunction}). Exploiting the magnetic translation algebra of the LLL states, we find~\cite{Note3}:
\begin{equation}
\Omega_{\bm k} = 2\sqrt{\det g_{\bm k}} = -1 +\Delta_{\bm k}\log\mathcal{N}_{\bm k},\label{expressionOmega}
\end{equation}
where $\Delta_{\bm k}$ is the Laplace operator. Eq.~(\ref{expressionOmega}) shows that the logarithm of the normalization factor $\mathcal{N}_{\bm k}$ is the $\bm k-$space K\"{a}hler potential \cite{Douglas:2009aa}, which controls the fluctuation of the band geometry and can be explicitly calculated \cite{Note3}:
\begin{equation}
\mathcal{N}_{\bm k}^{-2} = \sum_{\bm b}\eta_{\bm b}{w}_{\bm b}\exp\left(i\bm k\times\bm b\right)\exp\left(-\frac14|\bm b|^2\right),\label{BCNK}
\end{equation}
where $\eta_{\bm b}$=$+1$ if $\bm b/2$ is a reciprocal lattice vector and $-1$ otherwise, and $w_{\bm b}$ are the Fourier components of $|\mathcal{B}(\bm r)|^2$:
\begin{equation}
|\mathcal{B}(\bm r)|^2 = \sum_{\bm b}w_{\bm b}\exp(i\bm b\cdot\bm r),\label{FourierMode}
\end{equation}
where $\bm b=m_1\bm b_1+m_2\bm b_2$, $m_{1,2}\in\mathbb{Z}$ is a reciprocal lattice vector. The band geometry is uniform if $w_{\bm b\neq\bm 0}=0$.

\emph{Effective fractional quantum Hall model.---}A consequence of the exact wavefunction in Eq.~(\ref{modelwavefunction}) is that the interacting physics in a $\mC=1$ ideal flatband is described by a FQH-type model with a new Umklapp interaction that breaks continuous translation symmetry. We demonstrate that this Umklapp interaction captures precisely the fluctuating band geometry of an ideal flatband.

We consider a generic translation invariant two-particle interaction $v(\bm r_1-\bm r_2)$. According to Eq.~(\ref{modelwavefunction}), projecting this interaction into an ideal flatband yields an effective FQH model with the interaction
\begin{eqnarray}
&&\tilde{v}(\bm r_1,\bm r_2) = |\mathcal{B}(\bm r_1)\mathcal{B}(\bm r_2)|^2\cdot v(\bm r_1-\bm r_2),\label{effectiveinteractionQH}\\
&\approx&\sum_{\bm q}\!\!\left(\!\tilde{w}_0\!+\!\!\sum_{\bm b_i;j=1,2}\!\!(\tilde{w}_{i}e^{i\bm b_i\cdot\bm r_j} + h.c.)\!\right)\!v_{\bm q}e^{i\bm q(\bm r_1-\bm r_2)},\label{effectiveinteractionQH2}
\end{eqnarray}
projected to the LLL, 
where the normalization factors have been
%and normalizations can be 
dropped due to their weak $\bm k-$dependence according to Eq.~(\ref{BCNK}). The factor $|\mathcal{B}(\bm r)|^2$ reduces the continuous translation symmetry of $v(\bm r_1-\bm r_2)$ to the discrete lattice translation symmetry of $\tilde{v}(\bm r_1,\bm r_2)$. Such a symmetry reduction manifests itself as the inclusion of the ``Umklapp'' terms \footnote{We use ``Umklapp interaction'' and ``COM interaction'' inter-changeably.} that scatter electrons across the BZ which distinguish Eq.~(\ref{effectiveinteractionQH}) from the usual FQH models.

The effective FQH model Eq.~(\ref{effectiveinteractionQH}) can be simplified by retaining only the leading Umklapp processes that scatter electrons by the shortest distance in $\bm k-$space, because other Umklapp terms are suppressed after the LLL projection. This leads to Eq.~(\ref{effectiveinteractionQH2}) where the Umklapp interaction parameters $\tilde{w}_{0,1}$ can be easily derived from the Fourier modes $w_{\bm b}$ \cite{Note5}. To verify the validity of this approximation, we consider electrons at $1/3$ filling in the spin-valley polarized topological flatband of cTBG at the first magic angle. In this case, $\mathcal{B}(\bm r)$ is a two-component layer spinor $\left[i\mG(\bm r),\mG(-\bm r)\right]^T$ \cite{JieWang_NodalStructure}. The leading Umklapp processes scatter electrons by $\bm b_{1,2,3}$, shown in Fig.~\ref{illustration} with the same real amplitude $\tw_1$ due to the $\mC_3$ and exact intravalley inversion symmetries \cite{JieWang_NodalStructure}. In Fig.~\ref{illustration} (b), we plot the wavefunction's Fourier mode $w_{\bm b}$ of $|\mG(\bm r)|^2+|\mG(-\bm r)|^2$ and find $w_{\bm 0}, w_{\bm b_1}$ dominate, which determines the parameters in Eq.~(\ref{effectiveinteractionQH2}) to be $(\tw_0,\tw_1)=(1.35,0.3)$ \cite{Note5}. We assume electrons interact via a layer-isotropic $v_1$ Haldane pseudopotential $v(\bm r_1-\bm r_2)=\delta''(\bm r_1-\bm r_2)$ \cite{Haldane_hierarchy}. Remarkably, the entire low-energy spectrum of the cTBG model on the torus (blue dots in Fig.~\ref{compareu1}), including both the threefold degenerate ground states at zero energy and the gapped low-lying magnetoroton mode \cite{gmpl,PhysRevB.90.045114}, is well reproduced by the effective FQH model Eq.~(\ref{effectiveinteractionQH2}) with $\tw_{0}$ and $\tw_{1}$ [red crosses in Fig.~\ref{compareu1} (a)]. However, if we assume a uniform band geometry by setting $\tw_1=0$, the obtained spectrum [red crosses in Fig.~\ref{compareu1} (b)] shows significant deviations from the cTBG spectrum although the ground states stay at zero energy. This indicates our effective FQH model with nonzero $\tw_1$ indeed captures the spatially fluctuating band geometry of cTBG flatband.

\begin{figure}
    \centering
    \includegraphics[width=0.5\textwidth]{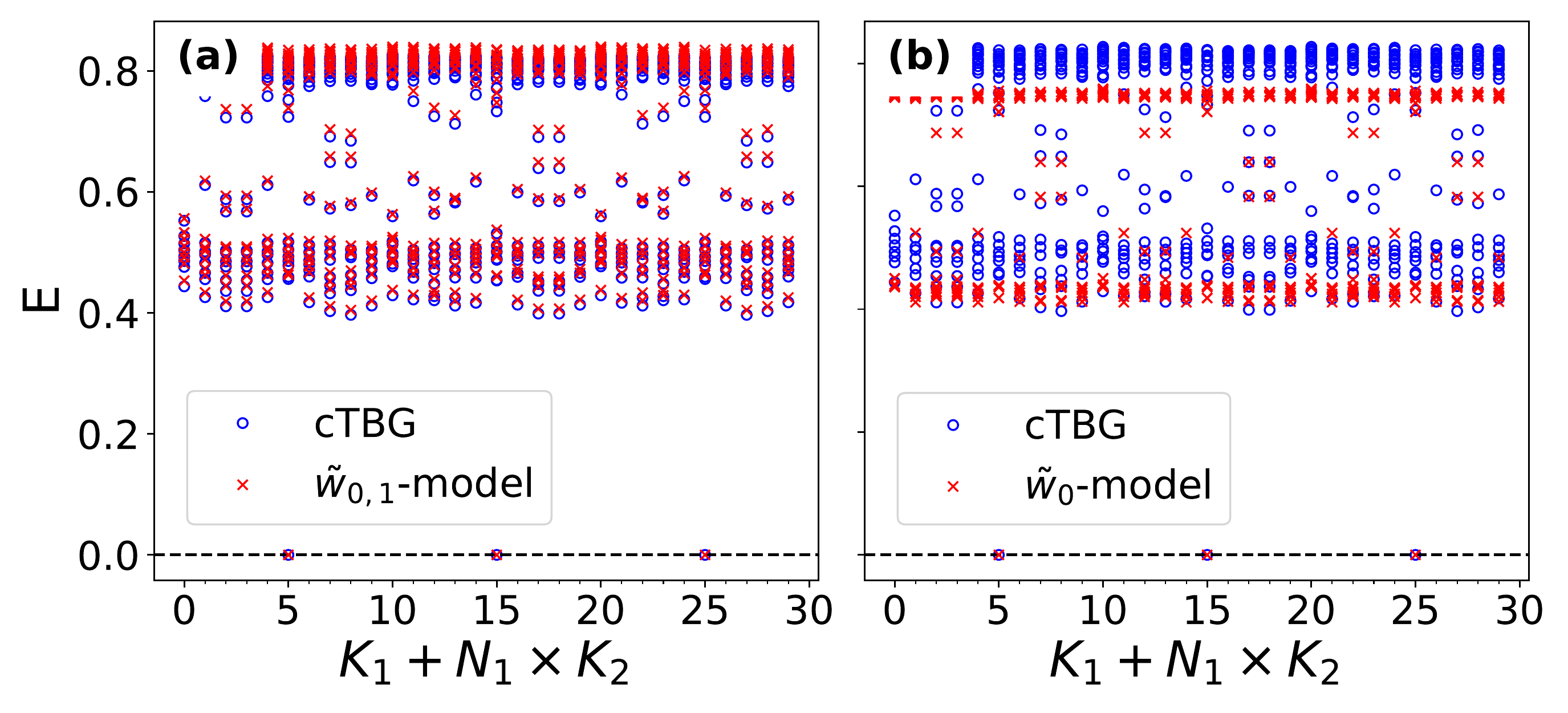}
    \caption{Exact diagonalization of $N=10$ particles in a TBG lattice of $(N_1,N_2)=(5,6)$ unit cells on the torus geometry, where $N_{1,2}$ are the number of unit cells along each primitive lattice direction. Many-body momenta $K_1\in[0,N_1-1]$, $K_2\in[0,N_2-1]$ are integers labeling each energy level \cite{haldanetorus1,haldanetorus2}. Blue circles are energies of cTBG with relative interaction $v_1$, and are the same in both panels. Red crosses are energies of the FQH model Eq.~(\ref{effectiveinteractionQH2}) with (a) $(\tw_0,\tw_1)=(1.35,0.3)$ and (b) $(\tw_0,\tw_1)=(1.35,0)$. Normalization factors Eq.~(\ref{BCNK}) are taken into account in numerical calculations. Including $\tw_1$ in (a) closely reproduces the low-energy details of the cTBG spectrum. Three exact zero modes are visible in both cTBG and the FQH models.}\label{compareu1}
\end{figure}

\emph{Center-of-mass pseudopotentials.---}The exact many-body zero modes observed above are the generalized-Laughlin states given in Ref.~(\onlinecite{Grisha_TBG2}), written here as $\Phi=\left(\prod_{i=1}^{N}\mathcal{B}(\bm r_i)\right)\Psi$, where $\Psi$ is the usual LLL Laughlin wavefunction. We now extend Haldane's pseudopotentials to capture the COM interactions. This allows us to systematically study how interactions can stabilize FQH states subject to nonuniform band geometry. We start by rewriting Eq.~(\ref{effectiveinteractionQH}) as follows:
\begin{eqnarray}
\tilde{v}(\bm r_1,\bm r_2) &=& \int d\bm q_+^2d\bm q_-^2\tilde{v}_{\bm q_+,\bm q_-}e^{i(\bm q_+\cdot\bm R^++\bm q_-\cdot\bm R^-)},\label{IntQq}\\
&=& \sum_{M,m}c_{M,m}\hat{P}^+_M\hat{P}^-_m, \label{IntMm}
\end{eqnarray}
where $\bm R^+$ and $\bm R^-$ are the LLL projected COM and relative coordinates of two particles. A generic two-particle interaction can be expressed in terms of its COM and relative translational momentum $\bm q_+$ and $\bm q_-$ as in Eq.~(\ref{IntQq}). For simplicity, we assume rotational symmetry, so that we can define projectors $\hat{P}_m^{\pm}\equiv2\int \frac{d^2\bm q}{(2\pi)^2} L_m(\bm q^2)e^{-\bm q^2/2}e^{i\bm q\cdot\bm R^{\pm}}$ which project the particle pair into its COM and relative angular momentum sectors respectively. The interacting Hamiltonian $\tilde{v}(\bm r_1,\bm r_2)$ can then be written as Eq.~(\ref{IntMm}), where $c_{M,m}=\int d\bm q_+^2d\bm q_-^2\tilde{v}_{\bm q_+,\bm q_-}L_M(\bm q_+^2)L_m(\bm q_-^2)$ is the generalized pseudopotential coefficient and $L_m$ is the Laguerre polynomial. The $c_{M,m}$ can be extracted from the energy spectrum of two interacting particles \cite{zhao_nonabelian,hierarchy_FCI}. 

The key insight here is that if $c_{M,1}>0$ and $c_{M,m>1}=0$, the generalized Laughlin state $\Phi$ has exactly zero energy, no matter how $c_{M,1}$ depends on the COM angular momentum $M$. We can thus construct a family of many-body states that are topologically equivalent to the Laughlin state, where the usual Laughlin state corresponds to the special case where $c_{M,1}$ is independent of $M$. Generalization to periodic lattice systems without rotational invariance is straightforward with generalized Laguerre polynomials \cite{generalizedPP_PRL,generalizedPP_PRB}. We emphasis that the statement is unchanged even if rotational invariance is broken: Eqs.~(\ref{effectiveinteractionQH}) and (\ref{effectiveinteractionQH2}) exhibit threefold exact zero modes at one-third filling for arbitrary orders of Umklapp scatterings of arbitrary strengths even with the $\bm k-$dependent normalization factors, as long as the relative interaction is the $v_1$ Haldane pseudopotential \cite{Note3}.

\emph{Center-of-mass interaction induced transitions.---}We now examine how the ground state and low energy physics of the effective FQH model in Eq.~(\ref{effectiveinteractionQH2}) evolve with $\tw_1/\tw_0$. We note that $\tw_1/\tw_0 $ is constrained by the band geometry; for example on a rectangular lattice with the $v_1$ interaction $|\tw_1/\tw_0|\leq0.25$ \footnote{The COM interaction parameters are: $\tw_0=w_{\bm 0}^2+4|w_{\bm b_1}|^2$, $\tw_1=w_{\bm 0}w_{\bm b_1}$ on square lattice; $\tw_0=w_{\bm 0}^2+6|w_{\bm b_1}|^2$, $\tw_1=w_{\bm 0}w_{\bm b_1}+w^{*2}_{\bm b_1}$ on triangular lattice. See SM for details}. However, sign changes in $\Omega_k$ or additional structure in the interaction may widen the allowed range. In Fig.~\ref{COMtransition}, we plot the ground state energies in unit of $\tw_0$ on the rectangular lattice as a function of $\tw_1/\tw_0$. We find that the ground state is the zero-energy generalized Laughlin state for small $|\tw_1/\tw_0|$. However for large enough $|\tw_1/\tw_0|$, the ground state energy becomes negative and the zero-energy generalized Laughlin state is an excited state. The occurrence of the negative-energy ground states can also be seen from the COM pseudopotentials plotted in Fig.~\ref{COMtransition} (b), where regions with negative values are shown.

To further understand the negative-energy ground states of the generalized FQH model, we compute the guiding-center structure factor $S(\bm q)$ $\equiv$ $\left(\langle\rho(\bm q)\rho(-\bm q)\rangle-\langle\rho(\bm q)\rangle\langle\rho(-\bm q)\rangle\right)/\left(N_1N_2\right)$, at $\tw_1/\tw_0$ = $0.3$ (before transition) and $0.5$ (after transition) for $N=12$ electrons in the $(N_1,N_2)=(6,6)$ lattice. The $S(\bm q)$ measures the density-density correlations of guiding centers and $\rho(\bm q)\equiv\exp(i\bm q\cdot\bm R)$ is the LLL-projected density operator. At $\tilde{w}_1/\tilde{w}_0=0.3$, the ground states are in the many-body momentum $K=(0,0)$ sector with exact zero energy. The corresponding structure factor has continuous peaks consistent with the incompressible Laughlin liquid. At $\tw_1/\tw_0=0.5$, the ground state is still in the $K=(0,0)$ sector, with nearby low-lying states at $\pm K_{\bm q}$ and $\pm\mC_4 K_{\bm q}$ where $K_{\bm q}=(3,0)$ and $\mC_4$ is the rotation by $\pi/2$. Remarkably, $S(\bm q)$ has discretized peaks exactly at $\pm K_{\bm q}$ and $\pm\mC_4 K_{\bm q}$. That the structure factor peak occurs exactly at momenta corresponding to low-energy excitations strongly suggests a gapless charge density wave (CDW) at $\tw_1/\tw_0=0.5$ \cite{KunYang_WignerCrystal}. The gapless CDW is analogous to the stripe phase and Wigner crystal reported in usual FQH systems at low filling factors. However in contrast to the usual FQH system, here the transition at a fixed filling factor $1/3$ is driven entirely by the band geometry and the transition occurs as a level crossing in Laughlin states' momentum sector $K=(0,0)$.

\begin{figure}
    \centering
    \includegraphics[width=0.5\textwidth]{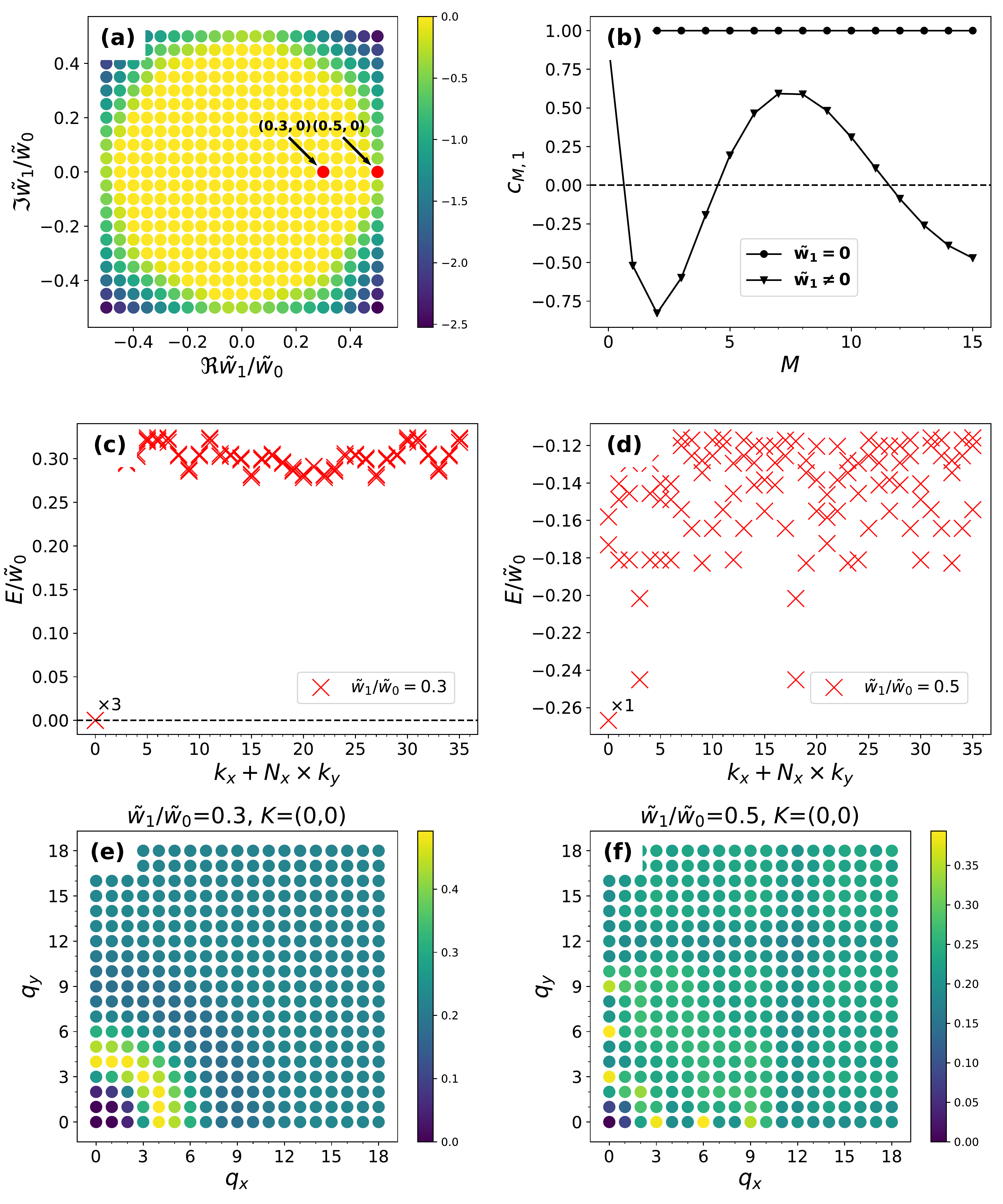}
    \caption{Center-of-mass interaction induced transitions. (a): Ground-state energies of the model in Eq.~(\ref{effectiveinteractionQH2}) as a function of $\tw_1/\tw_0$. Threefold exact zero modes are present for all values of $\tw_1/\tw_0$ (which remains true if $\tw_n$ is included to any order $n$). The appearance of negative-energy ground states is possible due to the negative value of the center-of-mass pseudopotential plotted in (b). Figures (c,d) and (e,f) are respectively the spectrum and the ground-state guiding-center structure factor $S(\bm q)$ for the two marked data points in (a) that represent typical phases before and after the transition. (c), (e) At $\tw_1/\tw_0 = 0.3$, the threefold degenerate zero-energy ground states, finite gap and continuous peak in $S(\bm q)$ are consistent with the Laughlin state. (d), (f) At $\tw_1/\tw_0=0.5$, the single negative-energy ground state, small excitation gap and discretized peaks in $S(\bm q)$ suggest a CDW phase. Plots are for a rectangular lattice on the torus, with $\bm b_1\cdot\bm b_2=0$ and $|\bm b_1|=|\bm b_2|$. The system sizes are $N=8$, $(N_1,N_2)=(4,6)$ in (a), and $N=12$, $(N_1,N_2)=(6,6)$ in (c)-(f).}\label{COMtransition}
\end{figure}

\emph{Discussion.---}We have studied interacting physics in ideal flatbands with inhomogeneous but constrained band geometries. Employing the exact correspondence to LLLs, we mapped the inhomogeneous band geometry in a flatband to a COM interaction in the LLL. Remarkably, as shown in FIG.~\ref{COMtransition} (b), the COM interaction generically has attractive components, driving a phase transition \cite{PhysRevB.103.125406} from the Laughlin state to gapless states.

The attractive interaction induced by band geometry implies new physics; applying COM pseudopotentials enables systematic studies of various instabilities which will be immediately useful for a wide range of applications \cite{XiDai_PseudoLandaulevel,popov2020hidden,Zaletel_TBG_AQH,ZhaoLiu_TBG,Cecile_TBG_Flatband,Cecile_PRL20,Senthil_NearlyFlatBand,Zaletel_PRX20,lian2020tbg,bernevig2020tbg_5,xie2020tbg,Nick_Spiralorder,PhysRevB.102.035158,PhysRevLett.126.137601,2020arXiv200311559K,Xu:2020aa,Regan:2020aa,Young_partiallyfilledTBG}. We conclude with two more examples. Recently, skyrmion pairing has been proposed to explain the superconductivity in TBG \cite{chargedskyrmion}, which was subsequently numerically tested in a simplified LLL based model with flat band geometries \cite{skyrmionSC_DMRG}. It is thus important to examine how inhomogeneous band geometry influences superconductivity. In the last section of the SM, we find the COM interaction induced by band geometry in time-reversal invariant TBG flatbands exhibits attractive components, which would favor superconductivity when its spatial pattern matches the superconducting order parameter. A thorough understanding requires extensive numerical studies that we leave for future work. The COM interaction is not only a novel concept but also a useful numerical tool, which we demonstrate through the second example by studying the stability of composite Fermi liquid (CFL) in the spin-valley polarized cTBG flatband. By continuously interpolating between the LLL and cTBG flatbands using the ideal flatband theory, we find that CFLs remain ground states of cTBG without signatures of phase transitions. Generalizations to higher Chern number and Hofstadter-type models are interesting future directions \cite{Kapit_Mueller,ZhaoLiu_TDBG,Bart_FQHmoireHofstadter,Bart_HighC,Bart_FQHtransition}.

\begin{acknowledgements}
\emph{Acknowledgements.---}
The Flatiron Institute is a division of the Simons Foundation. J. W. acknowledges Martin Claassen, Debanjan Chowdhury, and Kun Yang for useful discussions. J. C. acknowledges the support of the Air Force Office of Scientific Research under Grant No.~FA9550-20-1-0260. A. J. M. is supported in part by Programmable Quantum Materials, an Energy Frontier Research Center funded by the U.S. Department of Energy (DOE), Office of Science, Basic Energy Sciences (BES), under award DE-SC0019443. Z. L. is supported by the National Key Research and Development Program of China through Grant No. 2020YFA0309200. B. Y. acknowledges the support from the Singapore National Research Foundation (NRF) under NRF fellowship award NRF-NRFF12-2020-0005, a Nanyang Technological University start-up grant (NTU-SUG), and Singapore Ministry of Education MOE2018-T3-1-002.
\end{acknowledgements}

\bibliography{ref.bib}

\clearpage
\appendix
\section*{Supplementary material}
We include necessary theoretical and technical details in this supplementary material. It includes Section.~I on quantum Hall wavefunctions, with an emphasis on the $\bm k-$space holomorphic property of lowest Landau level wavefunctions, and how gauge transformations change the $\bm k-$space boundary condition; Section.~II on the ideal flatbands, which starts with detailed mathematical clarifications on the ideal flatband conditions used in the main text followed by the discussions on the uniqueness of our $\mC=1$ ideal flatband wavefunction; Section.~III on the interacting problems in ideal flatbands, which includes analytical derivations and numerical details on how to map flatband physics to Landau levels. Besides the phase transition discussed in the main text, we discuss two more applications about superconductivities and composite Fermi liquids in Section.~IV.

\tableofcontents

\section{Quantum Hall Wavefunctions}
In the first section, we review quantum Hall wavefunctions. We start with a gauge independent formalism to set up notations. After that, we discuss lowest Landau level (LLL) wavefunction in various gauges, and discuss their $\bm r-$space and $\bm k-$space holomorphic properties, and gauge transformations relating various gauges. The discussions about $\bm k-$space LLL wavefunctions and gauge transformations are important for our discussions on the ideal flatband wavefunctions in the next section.

\subsection{Gauge independent formalism}
The quantum Hall problem describes a two-dimensional electron gas in uniform perpendicular magnetic field. Under the minimal coupling, the gauge invariant momentum operator is:
\begin{equation}
\bm\pi_a \equiv -i\partial_a - e\bm A_a,
\end{equation}
where $\bm A$ is the vector gauge field of the external magnetic field whose curl is the magnetic field $\bm B = \bm\nabla\times\bm A$. In our notation, we set the electron charge $e<0$, and the $\hat{z}$ component of the magnetic field $B\equiv\hat{z}\cdot\bm B<0$, so we have $eB>0$. The magnetic length is defined as $l_B = \sqrt{\hbar/(eB)}$, which we implicitly set as \emph{one} as our length scale. The commutation relations of $\pi_a$ are,
\begin{equation}
~[\bm\pi_a, \bm\pi_b] = i(eB)\epsilon_{ab} = i\epsilon_{ab}l_B^{-2},\label{algebra_pi}
\end{equation}
where $\epsilon_{ab}$ is the 2D anti-symmetric tensor.

In magnetic field, the electron's motion is separated into guiding centers $\bm R$ and cyclotron motion $\bar{\bm R}$:
\begin{equation}
\bm R^a = \bm r^a + \epsilon^{ab}\bm\pi_bl_B^2,\quad\bm{\bar R}^a = -\epsilon^{ab}\bm \pi_b l_B^2,
\end{equation}
where $[\bm R^a, \bm{\bar R}^b]=0$, and:
\begin{equation}
~[\bm R^a, \bm R^b] = -i\epsilon^{ab}l_B^2,\quad [\bm{\bar R}^a, \bm{\bar R}^b] = i\epsilon^{ab}l_B^2,\label{algebra_R}
\end{equation}
which shows non-commutativity. Two sets of ladder operators are also defined:
\begin{equation}
a =\omega^*_a\bm R^a,\quad a^{\dag} = \omega_a\bm R^a,\quad\bar{a} = \omega_a\bar{\bm R}^a,\quad \bar{a}^{\dag} = \omega^*_a\bar{\bm R}^a,\label{ladders2}
\end{equation}
where $\omega_a$ is the \emph{complex structure} that determines the metric $g_{ab}$ and anti-symmetric tensor $\epsilon_{ab}$:
\begin{eqnarray}
g_{ab} &=& \omega_a^*\omega_b+\omega_a\omega_b^*,\nonumber\\
i\epsilon_{ab} &=& \omega_a^*\omega_b-\omega_a\omega_b^*,\label{defgepsilon}
\end{eqnarray}

The indices are raised or lowered by the metric. The complex structure satisfies $w_aw^a=0$, $w_a^*w^a=1$. The usual mathematical definition of complex structure is a two-tensor that square to minus one, which in fact corresponds to the $\epsilon^{ab}$ here; the ``complex structure'' used in this work should be understood as the null-vector of the projector $(g^{ab}+i\epsilon^{ab})/2$. The Einstein summation notation is assumed throughout this work.

In isotropic systems, we have $w_x=1/\sqrt{2}$, $w_y=i/\sqrt{2}$, leading to the usual representation $a=\left(R_x-iR_y\right)/\sqrt{2}$ and $\bar{a}=\left(\bar{R}_x+i\bar{R}_y\right)/\sqrt{2}$. We keep the discussion general by using Eq.~(\ref{ladders2}) with a general complex structure. The physical meaning of $\bar{a}/\bar{a}^{\dag}$ is the Landau level (LL) creation/annihilation operator, and $a/a^{\dag}$ means creates/annihilates an orbital (guiding center orbital) within a LL.

The quantum Hall system has a special aspect that the Berry curvature and Fubini-Study metric are uniform due to the continuous translation symmetry of the system in the uniform magnetic field. Thereby, any unit cells (where the aspect ratio and angle can vary) describes the problem equivalently as long as the primitive basis vectors of the unit cell $\bm a_1$ and $\bm a_2$ enclosed an area of $S=l_B^2=1$:
\begin{equation}
S = \frac{1}{2\pi}|\bm a_1\times\bm a_2|.\label{area}
\end{equation}

In solid state systems, however, there are natural definitions of unit cells from the underlying lattice, and the aspect ratio or angles cannot be varied continuously.

The reciprocal lattice vectors are defined as follows, such that $\bm a_i\cdot\bm b^j = 2\pi\delta^j_i$ is satisfied:
\begin{equation}
\bm b^i_a \equiv \epsilon^{ij}\epsilon_{ab}\bm a^b_j/S.\label{bvector}
\end{equation}

We assume the system is periodic in both directions: there are $N_1$ unit cells along $\bm a_1$ and $N_2$ unit cells along the $\bm a_2$ direction. So the whole system is defined on a torus with basis vectors:
\begin{equation}
\bm L_1 = N_1\bm a_1,\quad\bm L_2 = N_2\bm a_2.\label{Lvector}
\end{equation}

Since each unit cell has unit flux, the system has a total flux quanta:
\begin{equation}
N_{\phi}=N_1N_2.
\end{equation}

The reciprocal space is spanned by momentum points quantized on grids:
\begin{equation}
\bm q \in \{\frac{m_1}{N_1}\bm b^1+\frac{m_2}{N_2}\bm b^2|m_{1,2}\in\mathbb{Z}\}.\label{momentumquantization}
\end{equation}

We now discuss the magnetic translation symmetry in quantum Hall problems. The magnetic translation group elements are defined by the guiding center operators $\bm R$ as:
\begin{equation}
t(\bm q) \equiv \exp(i\bm q\cdot\bm R),\label{magtrans}
\end{equation}
which satisfy the commutation algebra:
\begin{equation}
t(\bm q_1)t(\bm q_2) = e^{i\bm q_1\times\bm q_2}t(\bm q_2)t(\bm q_1).\label{magalgebra}
\end{equation}

Translating a state across the Brillouin zone (BZ) must leave the state invariant. This defines the boundary condition:
\begin{equation}
e^{i\bm b^{i}\cdot\bm R}|\bm k\rangle = -e^{i\bm b^{i}\times\bm k}|\bm k\rangle,\quad i=1,2,\label{bc}
\end{equation}
where $\bm b^{1,2}$ are the two primitive reciprocal lattice vectors. This can be generalized to translation by multiple BZs. The phase factor $\exp(i\bm b\times\bm k)$ is analogous to the usual Bloch phase factor in solid states. Details of the derivations can be found in Ref.~(\onlinecite{haldanetorus1,Jie_MonteCarlo,haldanemodularinv}).

If $\bm q$ is not a reciprocal lattice vector, it translates a state $|\bm k\rangle$ to a different one $|\bm k+\bm q\rangle$:
\begin{equation}
e^{i\bm q\cdot\bm R}|\bm k\rangle = e^{\frac{i}{2}\bm q\times\bm k}|\bm k+\bm q\rangle.\label{symgaugerep}
\end{equation}

Taking $\bm q=\bm b^{i=1,2}$ and using the boundary condition Eq.~(\ref{bc}), we get a quasi-periodic boundary condition that any single-particle wavefunction in the Hilbert space needs to satisfy:
\begin{equation}
|\bm k+\bm b^i\rangle = -\exp(\frac{i}{2}\bm b^i\times\bm k)|\bm k\rangle.\label{ktranslation}
\end{equation}

Therefore, the states $|\bm k\rangle$, with the translation property Eq.~(\ref{symgaugerep}) and boundary condition Eq.~(\ref{ktranslation}) form a complete description of the Hilbert space of a single LL. While different Landau levels differ from each other by the details of the projection interaction, the essential physics of each Landau level is the same as they are all described by the $\bm R$ degrees of freedom. In this work, we use the LLL as a convenient tool to study the flatband problems.

\subsection{Lowest Landau level wavefunctions}
To discuss the LLL wavefunction in the first quantized form, it is necessary to fix the gauge first.

\subsubsection{Symmetric gauge}
We start with the symmetric gauge:
\begin{equation}
\bm A_a = -B\epsilon_{ab}\bm r^b/2,
\end{equation}
from which it follows that the Landau orbit and guiding center operators are:
\begin{equation}
\bar{\bm R}^a_{\sigma} = i\epsilon^{ab}\partial_b + \bm r^a/2 ,\quad\bm R^a_{\sigma} = -i\epsilon^{ab}\partial_b + \bm r^a/2.\label{RbarR}
\end{equation}

It is then straightforward to check that the wavefunction satisfying the magnetic translation symmetry in this symmetric gauge is given by the modified Weierstrass sigma function \cite{haldanemodularinv,Jie_MonteCarlo,FERRARI_Sigma,FERRARI_QHwannier}:
\begin{equation}
\Phi^{\sigma}_{\bm k}(\bm r) = \langle\bm r|\bm k\rangle = \sigma(z-z_k)e^{z_k^*z}e^{-\frac{1}{2}|z_k|^2}e^{-\frac{1}{2}|z|^2},\label{qhwf}
\end{equation}
where
\begin{equation}
z\equiv\omega_a\bm r^a,\quad z_k \equiv -i\omega^a\bm k_a,\label{defzzk}
\end{equation}
which are the complex coordinate of electrons and the complex coordinate of the zero respectively. The modified Weierstrass sigma function has a quasi-periodic translation symmetry:
\begin{equation}
\sigma(z+a_i)=-e^{a_i^*(z+a_i/2)}\sigma(z).\label{qbc_sigma}
\end{equation}

We will find it useful to define a $\bm k-$space holomorphic function:
\begin{eqnarray}
\tilde{u}^{\sigma}_{k}(\bm r) &\equiv& \Phi^{\sigma}_{\bm k}(\bm r)e^{-i\bm k\cdot\bm r}e^{\frac12|z_k|^2},\nonumber\\
&=& \sigma(z-z_k)e^{z^*z_k}e^{-\frac12|z|^2},\label{tildeusigma}
\end{eqnarray}
whose boundary condition is:
\begin{eqnarray}
\tilde{u}^{\sigma}_{k+b}(\bm r) &=& e^{i\phi_{k,b}}e^{-i\bm b\cdot\bm r}\tilde{u}^{\sigma}_{k}(\bm r),\label{kbc_u}
\end{eqnarray}
where
\begin{equation}
\phi_{k,b} = \pi+b^*(z_k-ib/2).\label{defphikb}
\end{equation}

\subsubsection{Landau gauge}
Landau gauge is another gauge that is frequently used. Without loss of generality, we choose $A_y=0$, $A_{x}=-By$. The Landau orbits and guiding centers are:
\begin{equation}
\bar{\bm R}^x_{\theta}=i\partial_y,\quad\bar{\bm R}^y_{\theta}=-i\partial_x+y,
\end{equation}
and
\begin{equation}
\bm R^x_{\theta}=x-i\partial_y,\quad \bm R^y_{\theta}=i\partial_x.
\end{equation}

% This Landau gauge can be regarded as a special limit of the symmetric gauge where $N_1=1$ and $N_2=N_{\phi}$. Since $N_1=1$, momenta have only a $\bm b_2$ component and are labeled by a single integer $|m\rangle\equiv|m\bm b_2/N_{\phi}\rangle$, see Eq.~(\ref{momentumquantization}). In this case, the translation symmetry in Eq.~(\ref{symgaugerep}) becomes:
% \begin{equation}
% e^{i\bm b_1\cdot\bm R}|m\rangle = e^{i\pi m/N_{\phi}}|m\rangle,\quad e^{i\bm b_2\cdot\bm R}|m\rangle = |m+1\rangle,
% \end{equation}
% and the boundary condition in Eq.~(\ref{ktranslation}) becomes:
% \begin{equation}
% |m+N_{\phi}\rangle=-|m\rangle.
% \end{equation}

Wavefunctions satisfying these translation properties and the boundary conditions are given by Jacobi theta functions which we label by superscript $\theta$,
\begin{equation}
\Phi^{\theta}_{\bm k}(\bm r) = e^{-z(z_k-z_k^*)}\theta_1(\frac{z-z_k}{a_1}|\tau)e^{-\frac{1}{2}k_x^2}e^{-\frac{1}{2}y^2},\label{thetawf}
\end{equation}
where the theta function is periodic in one direction and quasi-periodic in the other:
\begin{eqnarray}
\theta_1(z+1|\tau) &=& -\theta_1(z|\tau),\nonumber\\
\theta_1(z+\tau|\tau) &=& -e^{-i\pi\tau-2\pi i z}\theta_1(z|\tau),\label{thetatranslation}
\end{eqnarray}
and $\tau$ is the torus parameter $\tau=a_2/a_1$. Here, following the notation used in the main text, all unbolded $a_{i=1,2}\equiv\omega_a\bm a_i^a$ indicate complex coordinates.

\subsubsection{Gauge transformations}
We have denoted wavefunctions and operators in the symmetric gauge with superscript $\sigma$, and those in Landau gauge with superscript $\theta$. These two gauges are related by:
\begin{equation}
e^{-\frac{i}{2}xy}\cdot\bm R_{\sigma}\cdot e^{+\frac{i}{2}xy} = \bm R_{\theta},\quad e^{-\frac{i}{2}xy}\cdot\bm{\bar R}_{\sigma}\cdot e^{+\frac{i}{2}xy} = \bm{\bar R}_{\theta}.\label{gauge_sigma_theta}
\end{equation}

To see how this gauge transformation acts on wavefunctions, we use the mathematical relation between the sigma function and the theta function:
\begin{equation}
\sigma(z) = \frac{a_1}{\theta'_1}\theta_1(\frac{z}{a_1}|\tau)\exp\left(\frac{\eta_1}{a_1}z^2\right),\label{thetasigma}
\end{equation}
where $\eta_i$ is the Weierstrass Zeta function evaluated at half period $a_i/2$. We have $\eta_i=a_i^*/2$ in our case where $\bm a_{1,2}$ enclose one flux quanta. The $\theta'_1$ is the derivative of the theta function evaluated at the origin. Using Eq.~(\ref{thetasigma}) it is straightforward to show:
\begin{equation}
\Phi^{\sigma}_{\bm k}(\bm r) \sim \Phi^{\theta}_{\bm k}(\bm r)\times e^{-\frac{i}{2}k_xk_y+\frac{i}{2}xy},\label{PhithetaPhisigma}
\end{equation}
which is consistent with the gauge transformation in Eq.~(\ref{gauge_sigma_theta}). Here ``$\sim$'' means ``up to a constant''.

Analogous to $\tilde{u}^{\sigma}_{k}(\bm r)$ in Eq.~(\ref{tildeusigma}), we define the $\bm k-$holomorphic function in Landau gauge:
\begin{eqnarray}
\tilde{u}^{\theta}_{k}(\bm r) &=& \Phi^{\theta}_{\bm k}(\bm r)e^{-i\bm k\cdot\bm r}e^{\frac12k_x^2},\nonumber\\
&=& e^{-z_k(z-z^*)}\theta_1\left(\frac{z-z_k}{a_1}|\tau\right)e^{-\frac12y^2},\label{tildeutheta}
\end{eqnarray}
and by using Eq.~(\ref{thetatranslation}) and Eq.~(\ref{bvector}), we find its $\bm k-$space boundary condition:
\begin{equation}
\tilde{u}^{\theta}_{k+b_i}(\bm r) = e^{\phi'_{k,b_i}-i\bm b_i\cdot\bm r}\tilde{u}^{\theta}_{k}(\bm r),\quad i=1,2,
\end{equation}
where
\begin{eqnarray}
\phi'_{k,b_1} &=& \pi + (b_1^*+b_1)z_k + \sum_{i=1}^2b^*_i(-ib_i/2),\nonumber\\
\phi'_{k,b_2} &=& \pi.\label{defphikb_prime}
\end{eqnarray}

We have shown that the commonly used two torus LLL wavefunctions ($\theta$ and $\sigma$ functions) are related by a gauge transformation. And interestingly, the $\bm k-$space holomorphic function $\tilde{u}^{\sigma}_k$ and $\tilde{u}^{\theta}_k$ in Eq.~(\ref{tildeusigma}) and Eq.~(\ref{tildeutheta}) has the same form of their real space wavefunction with $z\leftrightarrow z_k$ interchanged, resembling the space-momentum duality in the LLL. In flatbands, while the real space wavefunction lost its LLL holomorphic properties, the $\bm k-$space holomorphic properties appear under our ideal flatband condition as we will discuss in detail in the next section. We will also find the $\bm k-$space boundary condition Eq.~(\ref{defphikb}) and Eq.~(\ref{defphikb_prime}) useful in determining our $\mC=1$ ideal flatband wavefunction.

We next briefly discuss the general gauge, which interpolates between the symmetric and Landau gauges parameterized by $\kappa$,
\begin{eqnarray}
e^{-\frac{i}{2}\kappa xy}\cdot\bm R_{\sigma}\cdot e^{+\frac{i}{2}\kappa xy} &=& \bm R_{\kappa},\nonumber\\
e^{-\frac{i}{2}\kappa xy}\cdot\bar{\bm R}_{\sigma}\cdot e^{+\frac{i}{2}\kappa xy} &=& \bar{\bm R}_{\kappa}.\label{generalRbarR}
\end{eqnarray}

We denote operators and wavefunctions under this gauge by a superscript $\kappa$. The choice of $\kappa=0$ and $\kappa=-1$ recover the symmetric and Landau gauge respectively, but general $\kappa$ gives us different boundary conditions. To solve for wavefunctions in general gauge, we first define a new elliptic function,
\begin{equation}
\sigma_{\kappa}(z) \equiv e^{-\kappa\frac{\eta_1}{2\omega_1}z^2}\sigma(z),\label{sigmakappa}
\end{equation}
which, same as $\sigma(z)$ and $\theta_1(z|\tau)$, $\sigma_{\kappa}(z)$ is holomorphic and quasi-periodic:
\begin{equation}
\sigma_{\kappa}(z+a_i) = -\sigma_k(z)\cdot e^{(a^*_i-\kappa a_i)(z+\frac{1}{2}a_1)}.
\end{equation}

The Eq.~(\ref{sigmakappa}) is motivated by the relation between $\sigma(z)$ and $\theta_1(z)$ as shown in Eq.~(\ref{thetasigma}). One can always choose $a_1$ to be real to achieve a simpler expression:
\begin{equation}
\sigma_{\kappa}(z)=\sigma(z)\exp\left(-\frac{\kappa}{2}z^2\right).
\end{equation}

The wavefunction that respects translation symmetry Eq.~(\ref{symgaugerep}) and the boundary condition Eq.~(\ref{ktranslation}) are given by:
\begin{eqnarray}
\Phi^{\kappa}_{\bm k}(\bm r) &=& \sigma_{\kappa}(z-z_k) e^{(z^*_k-\kappa z_k)z}\\
&\times& e^{-\frac{1}{2}|z|^2 + \frac{\kappa}{4}(z^2+\bar z^2)} e^{-\frac{1}{2}|z_k|^2 + \frac{\kappa}{4}(z_k^2+\bar z_k^2)},\nonumber
\end{eqnarray}
and it is related to the symmetric gauge wavefunction by:
\begin{equation}
\Phi^{\sigma}_{\bm k}(\bm r) \sim \Phi^{\kappa}_{\bm k}(\bm r)\times e^{i\frac{\kappa}{2}(k_xk_y-xy)}.\label{PhikappaPhisigma}
\end{equation}

Same as the discussions for $\tilde{u}^{\sigma}_k$ and $\tilde{u}^{\theta}_k$, one can define a $\bm k-$space holomorphic function $\tilde{u}^{\kappa}_k$ from Eq.~(\ref{PhikappaPhisigma}). Its $\bm k-$space boundary condition is a linear interpolation between Eq.~(\ref{defphikb}) and Eq.~(\ref{defphikb_prime}). We will find this useful in determining the $\mC=1$ ideal flatband wavefunction.

\section{Band geometry and wavefunctions}
In this section, we provide details of various mathematical statements used in the main text, and show the uniqueness of our model wavefunction.

\subsection{Band geometry and k-space complex structure}
Following the notation used in the main text: we denote the Bloch function as $\psi_{\bm k}(\bm r)$ and denote its cell-periodic part as $u_{\bm k}(\bm r)$:
\begin{equation}
\psi_{\bm k}(\bm r) = u_{\bm k}(\bm r)\exp(i\bm k\cdot\bm r),
\end{equation}
where both $\psi_{\bm k}$ and $u_{\bm k}$ are normalized. The quantum geometric tensor is defined as:
\beqn
\mathcal{Q}^{ab}(\bm k) &=& \langle D^au_{\bm k}|D^bu_{\bm k}\rangle,\nonumber\\
|D^au_{\bm k}\rangle &=& (\partial_k^a - iA^a_{\bm k})|u_{\bm k}\rangle,
\eeqn
where $|D^au_{\bm k}\rangle$ is the covariant derivative that satisfies,
\begin{equation}
\langle u_{\bm k}|D^au_{\bm k}\rangle = 0,\label{defcovD}
\end{equation}
and the Berry connection and Berry curvature are defined respectively as follows:
\begin{equation}
A^a_{\bm k} \equiv -i\langle u_{\bm k}|\partial^a_{\bm k}u_{\bm k}\rangle,\quad\Omega_{\bm k}=\epsilon_{ab}\partial^a_{\bm k}A^b_{\bm k},\label{defAk}
\end{equation}

The imaginary and real part of $\mathcal{Q}^{ab}_{\bm k}$ are the Berry curvature and the Fubini-Study metric:
\begin{equation}
\mathcal{Q}^{ab}_{\bm k} = g_{\bm k}^{ab} + \frac{i\epsilon^{ab}}{2}\Omega_{\bm k}.
\end{equation}

The geometry of a single-band Bloch wavefunction is described by the mapping from the 2D Brillouin zone to the Bloch wavefunction: $\text{BZ}\rightarrow\text{CP}_n$, where $n$ is the dimension of the Hamiltonian. Mathematically, the Fubini-Study metric $g_{\bm k}^{ab}d\bm k_ad\bm k_b$ is the pull-back of the standard $CP^n$ metric to the 2D Brillouin zone.

The definition of the Berry connection in Eq.~(\ref{defAk}) differs from the usual definition by a minus sign: as we will see this notation gives us ``holomorphic'' wavefunctions with ``positive'' Berry curvature (the usual notation gives ``holomorphic'' wavefunctions with ``negative'' Berry curvature). Note that the definition of covariant derivative is independent on the sign choice as Eq.~(\ref{defcovD}) has to be satisfied. This Berry connection sign convention has also been used in Ref.~(\onlinecite{haldaneanomaloushall}).

It is well known in the literature that \cite{PARAMESWARAN2013816,PhysRevB.90.165139,Jackson:2015aa,Martin_PositionMomentumDuality}:
\begin{theorem}
The Fubini-Study metric and Berry curvature satisfy the following inequalities:
\begin{equation}
|\frac{1}{2}\omega_{\bm k,ab}g_{\bm k}^{ab}|\geq\sqrt{\det g_{\bm k}}\geq\frac{1}{2}|\Omega_{\bm k}|.\label{inequality}
\end{equation}
where $\omega_{\bm k}^{ab}$ is an arbitrary unimodular matrix, that may vary at different $\bm k$.
\end{theorem}
\begin{proof}
Introducing the band projector
\begin{equation}
P_{\bm k} = |u_{\bm k}\rangle\langle u_{\bm k}|,\quad Q_{\bm k} = \mathbb{I}-P_{\bm k},
\end{equation}
with which the quantum geometric tensor can be rewritten as
\begin{equation}
\mathcal{Q}^{ab}_{\bm k} = 2\langle\partial^a_{\bm k}u|Q_{\bm k}|\partial^b_{\bm k}u\rangle.\label{def_QGC}
\end{equation}
from which it is easy to see it is positive semi-definite. Since $\det\mathcal{Q}_{\bm k} = \det g_{\bm k} - \Omega_{\bm k}^2/4\geq0$, we proved the second inequality of Eq.~(\ref{inequality}).

The first inequality follows from the inequality between trace and determinant for general hermitian matrices which of course applies to the real symmetric matrix $g_{\bm k}^{ab}$ here. To see this more clearly, we follow Ref.~(\onlinecite{PhysRevB.90.165139}) to transform into the diagonal coordinate system of $\omega_{\bm k,ab}$. In this coordinate system, write $\omega_{\bm k}$ and $g_{\bm k}$ as follows,
\begin{equation}
\omega_{\bm k} = \left(\begin{matrix}\lambda_{\bm k}&0\\0&\lambda^{-1}_{\bm k}\end{matrix}\right),\quad g_{\bm k} = \sqrt{\det g_{\bm k}}\left(\begin{matrix}\tilde{g}_{11}&\tilde{g}_{12}\\\tilde{g}_{21}&\tilde{g}_{22}\end{matrix}\right).
\end{equation}
where $\tilde{g}$ is the unimodular part of $g_{\bm k}$. Then the contraction between $\omega_{\bm k}$ and $g_{\bm k}$ equals to $\sqrt{\det g_{\bm k}}(\lambda_{\bm k}\tilde{g}_{11}+\lambda^{-1}_{\bm k}\tilde{g}_{22})/2\geq\sqrt{\det g_{\bm k}}\sqrt{\tilde{g}_{11}\tilde{g}_{22}} = \sqrt{\det g_{\bm k}}\sqrt{1+\tilde{g}_{12}^2}$ which is greater than $\sqrt{\det g_{\bm k}}$.
\end{proof}

If we replace $\omega^{ab}_{\bm k}$ by the $\delta^{ab}$ matrix in Eq.~(\ref{inequality}), the left hand side of this inequality becomes the trace of the Fubini-Study metric. The bound on the trace of $g^{ab}_{\bm k}$ in terms of the Berry curvature was initially noticed in Ref.~(\onlinecite{PhysRevB.90.165139}) and was termed by the ``trace condition''. Saturation of the trace bound implies the saturation of the determinant bound.

We now examine what happens when the determinant bound saturates.

\begin{corollary}\label{coro_bound}
The bound
\begin{equation}
\sqrt{\det g_{\bm k}}\geq\frac{1}{2}|\Omega_{\bm k}|,\label{detbound_SM}
\end{equation}
saturates if and only if $\sqrt{\det g_{\bm k}}=\Omega_{\bm k}/2=0$ or $\mathcal{Q}^{ab}_{\bm k}$ has a null vector $\omega_{\bm k,a}$:
\begin{equation}
\mathcal{Q}^{ab}_{\bm k}\omega_{\bm k,b} = 0.\label{defnullvec}
\end{equation}

If $\Omega_{\bm k}\neq0$ (without loss of generality we assume $\Omega_{\bm k}>0$) and Eq.~(\ref{detbound_SM}) saturates, then the Berry curvature and the Fubini-Study metric can be written as:
\begin{eqnarray}
i\epsilon^{ab}\frac{\Omega_{\bm k}}{|\Omega_{\bm k}|} &=& \omega^{a*}_{\bm k}\omega^{b}_{\bm k} - \omega^{a}_{\bm k}\omega^{b*}_{\bm k},\nonumber\\
\frac{g_{\bm k}^{ab}}{\sqrt{\det g_{\bm k}}} &=& \omega^{a*}_{\bm k}\omega^{b}_{\bm k} + \omega^{a}_{\bm k}\omega^{b*}_{\bm k}.\label{omegagasw}
\end{eqnarray}
\end{corollary}

\begin{proof}
If $\Omega_{\bm k}$ and all components of $g_{\bm k}$ vanish identically, of course $\sqrt{\det g_{\bm k}}=\Omega_{\bm k}/2$ holds. Now we discuss the case with $\Omega_{\bm k}\neq0$.

Because the spectrum of $\mathcal{Q}^{ab}_{\bm k}$ is non-negative, saturation of the bound implies at least one zero eigenvalue. Therefore there must exist a $\omega_{\bm k,a}$ such that Eq.~(\ref{defnullvec}) is true. Eq.~(\ref{omegagasw}) follows from Eq.~(\ref{defgepsilon}) by assuming $\Omega_{\bm k}>0$.
\end{proof}

The determinant bound is trivially saturated in any 2D two-band models. However there is a caveat:
\begin{theorem}
Any two-dimensional two-band model saturates the bound Eq.~(\ref{detbound_SM}). However, for two-dimensional two-band model, there must exist a point $\bm k_0$ in the Brillouin zone (BZ) such that $\det g_{\bm k}=\Omega_{\bm k}^2/4=0$.
\end{theorem}
\begin{proof}
The first statement is easy to prove. The $\Omega_{\bm k}$ and $g^{ab}_{\bm k}$ can be expressed by the multi-band Berry connection. For two-band models, they are:
\begin{eqnarray}
\epsilon^{ab}\Omega_{\bm k} &=& i\tilde{A}^a_{\bm k}\tilde{A}^{b*}_{\bm k} - i\tilde{A}^{a*}_{\bm k}\tilde{A}^b_{\bm k},\nonumber\\
g^{ab}_{\bm k} &=& \tilde{A}^a_{\bm k}\tilde{A}^{b*}_{\bm k} + \tilde{A}^{a*}_{\bm k}\tilde{A}^b_{\bm k}.
\end{eqnarray}
where $\tilde{A}_{\bm k}^a\equiv-i\langle u_{0\bm k}|\partial_{\bm k}^au_{1\bm k}\rangle$ is the inter-band Berry connection between the band of interest ($n=0$) and the other one. It follows that,
\begin{eqnarray}
\det g_{\bm k} &=& \frac{1}{2}\epsilon_{ac}\epsilon_{bd} g_{\bm k}^{ab}g_{\bm k}^{cd},\nonumber\\
&=& \frac{1}{2}\left(\epsilon_{ac}\tilde{A}^a_{\bm k}\tilde{A}^{c*}_{\bm k}\right)\left(\epsilon_{bd}\tilde{A}^{b*}_{\bm k}\tilde{A}^d_{\bm k}\right)+h.c.,\nonumber\\
&=& \Omega_{\bm k}^2/4.
\end{eqnarray}

The proof for the second statement that $\det g_{\bm k}=\Omega_{\bm k}$ must vanish in the BZ can be found in Refs.~(\onlinecite{kahlerband1,kahlerband2}).
\end{proof}

Saturation of the determinant bound with nonzero Berry curvature also implies the local complex structure $\omega_{\bm k,a}$ satisfies Eq.~(\ref{omegagasw}). However, in order for the local complex coordinate system to be able to be smoothly glued together to form a global coordinate system, we must impose $\Omega_{\bm k}$ positive (or negative) definite so that it does not change sign in the BZ. Following Refs.~(\onlinecite{kahlerband1,kahlerband2}):
\begin{theorem}\label{theorm_holocoord}
If $\det g_{\bm k}\neq0$ and the bound in Eq.~(\ref{detbound_SM}) saturates for all $\bm k$, there exists a smooth function $\lambda_{\bm k}$ such that for all $\bm k$:
\begin{equation}
Q_{\bm k}\frac{\partial P_{\bm k}}{\partial k_1} = \lambda_{\bm k}Q_{\bm k}\frac{\partial P_{\bm k}}{\partial k_2}.
\end{equation}
In other words, in terms of a local holomorphic coordinate $\partial_{\bar z}\equiv\left(\partial_{k_1} - \lambda_{\bm k}\partial_{k_2}\right)/2$, we have,
\begin{equation}
Q_{\bm k}\frac{\partial P_{\bm k}}{\partial \bar z} = 0.\label{holo_coord}
\end{equation}
\end{theorem}
\begin{proof}
See Refs.~(\onlinecite{kahlerband1,kahlerband2}) for details. Discussions there also include cases described by a single momentum which correspond to the Jacobi theta function representation discussed earlier.
\end{proof}

Therefore, we derived the necessary and sufficient condition for the existence of $\bm k-$space local complex coordinate: the saturation of the determinant bound Eq.~(\ref{detbound_SM}) and $\Omega_{\bm k}\neq0$ at that $\bm k$. When these two conditions are satisfied, the local complex structure is given by the null vector of the quantum geometric tensor $\mathcal{Q}^{ab}_{\bm k}\omega_{\bm k,b}=0$, and we have $k\equiv \omega^a_{\bm k}\bm k_a$. In order for the local complex coordinate systems to be glued together to form a global complex coordinate system, we demand that the Berry curvature is positive (or negative) definite and the determinant bound is saturated for every $\bm k-$point in the BZ. However, mapping to the usual Landau levels of uniform magnetic fields is still a highly nontrivial problem since the local complex structure $\omega_{\bm k,a}$ has $\bm k-$dependence, giving raise to nontrivial $\bm k-$space curvature. Therefore, we assume a stronger condition: the $\omega_{\bm k,a}=\omega_a$ is $\bm k-$independent, which is equivalent as assuming the existence of a constant determinant one matrix $\omega^{ab}=\omega^{a*}\omega^b+\omega^a\omega^{b*}$ such that $g^{ab}_{\bm k}=\omega^{ab}\Omega_{\bm k}/2$ is true everywhere in the BZ. To conclude, this justifies the physical and mathematical motivation for the ideal conditions proposed as (i) and (ii) in the introduction part of the main text.

As noticed initially in Ref.~(\onlinecite{Martin_PositionMomentumDuality}), the wavefunction of an ideal flatbands have a $\bm k-$space holomorphic property:

\begin{theorem}
A single band in 2D with positive (or negative) Berry curvature whose quantum geometric tensor has a constant null vector $\omega^a$, if and only if the wavefunction is:
\begin{equation}
u_{\bm k}(\bm r) = N_{\bm k}\tilde{u}_k(\bm r),\label{holo_wf}
\end{equation}
{\it i.e.} the cell-periodic part of the Bloch wavefunction can be written as a holomorphic function of $k\equiv\omega^ak_a$ up to a normalization factor.
\end{theorem}
\begin{proof}
We first review the proof that wavefunction in Eq.~(\ref{holo_wf}) saturates the bound noticed in Ref.~(\onlinecite{Martin_PositionMomentumDuality}). The quantum geometric tensor can be expressed as the follows for any unnormalized cell-periodic function \cite{Grisha_TBG2}:
\begin{equation}
\mathcal{Q}^{ab}_{\bm k} = {N}^2_{\bm k}\left(\langle\partial_{\bm k}^a\tilde{u}_{\bm k}|\partial_{\bm k}^b\tilde{u}_{\bm k}\rangle\right) - {N}^4_{\bm k}\left(\langle\partial_{\bm k}^a\tilde{u}_{\bm k}|\tilde{u}_{\bm k}\rangle\langle \tilde{u}_{\bm k}|\partial_{\bm k}^b\tilde{u}_{\bm k}\rangle\right).\label{def_Q_unnorm}
\end{equation}

If $\tilde{u}_{\bm k}$ is holomorphic in $k$, by definition we have $\omega_a \partial^a_{\bm k}|\tilde u_{\bm k}\rangle = 0$. Contracting this complex structure with the quantum geometric tensor gives:
\begin{equation}
\omega^*_a\mathcal{Q}^{ab}_{\bm k} = \mathcal{Q}_{\bm k}^{ab}\omega_b = 0.
\end{equation}

To see the inverse, note that if the band has positive (or negative) Berry curvature and constant complex structure, then we have $Q_{\bm k}\partial_{\bar k}P_{\bm k}=0$ following Theorem~\ref{theorm_holocoord}. In a single band, this means $\partial_{\bar k}P_{\bm k}$ must be proportional to $P_{\bm k}$ itself, {\it i.e.} $\partial_{\bar k}|u_{\bm k}\rangle=\lambda_{\bm k}|u_{\bm k}\rangle$ for some prefactor $\lambda_{\bm k}$. If $\lambda_{\bm k}=0$ we finished the proof and $u_{\bm k} = \tilde{u}_k$. If not, one can adjust the normalization $N_{\bm k}$ and the wavefunction's phase factor $\phi_{\bm k}$ to have $\lambda_{\bm k}\equiv\partial_{\bar k}\log\left(N_{\bm k}\exp(i\phi_{\bm k})\right)$, such that $\tilde{u}_k\equiv u_{\bm k}/\left(N_{\bm k}\exp(i\phi_{\bm k})\right)$.
\end{proof}

We conclude this section by collecting the results discussed in this section:
\begin{theorem}\label{ideal_property}
For a single band in 2D, if for all $\bm k$ the following conditions are satisfied: (i) $\Omega_{\bm k}\neq0$, (ii) $\exists~\omega_a\neq0$ {\it s.t.} $\mathcal{Q}^{ab}_{\bm k}\omega_b=0$, then its cell-periodic wavefunction can be chosen as a function that is holomorphic in $k$ as shown in Eq.~(\ref{holo_wf}), where the complex coordinate $k\equiv\omega^a\bm k_a$. Furthermore, we have:
\begin{equation}
g_{\bm k}^{ab} = \frac{1}{2}\left(\omega^a\omega^{b*} + \omega^{a*}\omega^b\right)\Omega_{\bm k},
\end{equation}
which implies $\det g_{\bm k}=\Omega^2_{\bm k}/4$ automatically.
\end{theorem}

\subsection{Wavefunction of C=1 ideal flatbands}
We have clarified the ideal flatband conditions. In this section we derive the general form of the ideal flatband wavefunction when the Chern number $\mC=1$.
\begin{theorem}\label{lemma1}
For an ideal band with non-zero Chern number, the normalization factor cannot be periodic if the wavefunction is smooth in the BZ:
\begin{equation}
N_{\bm k}\neq N_{\bm k+\bm b}.
\end{equation}
where $\bm b$ is a reciprocal lattice vector. Here $N_{\bm k}$ is defined in Eq.~(\ref{holo_wf}).
\end{theorem}
\begin{proof}
We prove this by contradiction. We assume $N_{\bm k}$ is periodic and show it necessarily implies $C=0$. As in the main text, we define $\tilde\phi_{\bm k,\bm b}$ as the phase factor when the Bloch function is translated in $\bm k-$space by a reciprocal lattice vector:
\begin{equation}
\psi_{\bm k+\bm b}(\bm r) = e^{i\tilde\phi_{\bm k,\bm b}}\psi_{\bm k}(\bm r).\label{psikbcd}
\end{equation}

Since $\psi_{\bm k}$ is normalized, $\tilde\phi_{\bm k,\bm b}$ is a real phase factor. The boundary condition of $\tilde{u}_{\bm k}$ is straightforwardly obtained by assuming $N_{\bm k}$ periodic:
\begin{equation}
\frac{\tilde{u}_{k+b}(\bm r)}{\tilde{u}_{k}(\bm r)}e^{-i\tilde\phi_{\bm k,\bm b}} = e^{-i\bm b\cdot\bm r}.\label{bctildeu}
\end{equation}

Since both $\tilde{u}_{k}$ and $\tilde{u}_{k+b}$ are holomorphic in $k$, phase $\tilde\phi_{\bm k,\bm b}$ must be a holomorphic function of $k$ as well. Then $\tilde\phi_{\bm k,\bm b}$ must be a constant because it is also real. However, constant boundary condition of Bloch function is inconsistent with non-zero Chern number suppose the Bloch function is smooth in the BZ bulk (we implicitly assume this is true). Therefore we reached a contradiction.
\end{proof}

For comparison, Ref.~(\onlinecite{Martin_PositionMomentumDuality}) considered periodic wavefunctions. Periodic holomorphic functions are described by the Weierstrass elliptic functions with at least a second order poles in the Brillouin zone bulk. Therefore discussions there were constrained for flatbands with Chern number $C\geq2$.

Theorem~\ref{lemma1} shows that we need to replace the real phase factor $\tilde\phi$ of Eq.~(\ref{bctildeu}) with a complex phase factor $\Im\phi_{\bm k,\bm b}\neq0$ when $C\neq0$:
\begin{equation}
\frac{\tilde{u}_{k+b}(\bm r)}{\tilde{u}_{k}(\bm r)}e^{-i\phi_{\bm k,\bm b}} = e^{-i\bm b\cdot\bm r},\quad\bar{\partial}_k\phi_{\bm k,\bm b} = 0.\label{boundcondtildeu}
\end{equation}
where the nonzero imaginary part of $\phi_{\bm k,\bm b}$ determines how $N_{\bm k}$ decays when translating in $\bm k-$space:
\begin{equation}
N_{\bm k+\bm b}/N_{\bm k} = \exp\left(\Im\phi_{\bm k,\bm b}\right),
\end{equation}
and the real part of $\phi_{\bm k,\bm b}$ gives to the boundary condition of the Bloch function:
\begin{equation}
\Re{\phi_{\bm k,\bm b}} = \tilde{\phi}_{\bm k,\bm b}.
\end{equation}

Since the complex phase $\phi_{\bm k,\bm b}$ is holomorphic in momentum, we wrote it as $\phi_{k,b}$ from now on. The $\phi_{k,-b}$ is not independent from $\phi_{k,b}$: by inverting Eq.~(\ref{boundcondtildeu}), we find:
\begin{equation}
\phi_{k+b,-b} = -\phi_{k, b},\mod2\pi.\label{phiinvertk}
\end{equation}

We then seek a relation between Chern number $\mC$ and the complex phase $\phi_{k, b}$. To do this, we first define a quantity $\mC'(\bm r)$:
\begin{equation}
\mC'(\bm r) = \frac{1}{2\pi i}\oint dk~\partial_k\ln\tilde{u}_{k}(\bm r),\label{boundaryintzero}
\end{equation}
which measures the phase winding of $\tilde{u}_k$ in the BZ. By the argument principle of complex analysis, we know $\mC'(\bm r)$ must be an integer $\mC'(\bm r)\in\mathbb{Z}$ for any fixed position $\bm r$. We now prove this integer is the Chern number:
\begin{theorem}\label{lemma2}
The $C'(\bm r)$ is independent on $\bm r$ and equals to the Chern number $\mC$.
\end{theorem}
\begin{proof}
The value of $\mC'(\bm r)$ can be explicitly calculated using the integral contour shown in the main text:
\begin{equation}
\mC'(\bm r) = \frac{-1}{2\pi}\left(\phi_{k_0+b_1, b_2} - \phi_{k_0, b_2} + \phi_{k_0, b_1} - \phi_{k_0 + b_2, b_1}\right),\label{quantizedphase}
\end{equation}
where $k_0$ is the origin of the BZ, and equivalently the origin of the integral contour. From Eq.~(\ref{quantizedphase}), we find $\mC'(\bm r)$ is a constant independent on the parameter $\bm r$. Thereby we get $\mC'=\mC'(\bm r)$. To show $\mC'$ equals to the Chern number, we consider check the periodicity of the Berry connection:
\begin{equation}
A^a_{\bm k+\bm b} = A^a_{\bm k} - \partial_{\bm k}^a\tilde\phi_{\bm k,\bm b} = A^a_{\bm k} - \partial_{\bm k}^a\Re{\phi}_{k, b},
\end{equation}
where we used $\tilde{\phi}_{k,b}=\Re\phi_{k,b}$.

For smooth wavefunctions, the Chern number is given by the BZ boundary integral $\mC=\frac{1}{2\pi}\oint d\bm k_a\bm A^a_{\bm k}$. Performing Brillouin zone boundary integration, we find $\mC$ as:
\begin{equation}
\mC = -\frac{1}{2\pi}\Re\left(\phi_{k_0+b_1, b_2} - \phi_{k_0, b_2} + \phi_{k_0, b_1} - \phi_{k_0 + b_2, b_1}\right).\label{ck0}
\end{equation}

Comparing Eq.~(\ref{quantizedphase}) with Eq.~(\ref{ck0}), we find,
\begin{equation}
C=\Re C' = C'.
\end{equation}
\end{proof}

We now show how the complex phase $\phi_{k, b}$ is constrained by the Chern number:
\begin{theorem}\label{lemma3}
The complex phase factor $\phi_{k, b}$ must be a linear function of $k$, constrained by the Chern number. Without loss of generality, the boundary condition can be chosen as follows:
\begin{equation}
\phi_{k,b}=b^*(-ik-ib/2)+\pi.\label{complexphi}
\end{equation}
\end{theorem}
\begin{proof}
We expand the complex phase in terms of the holomorphic $k$. For notational consistent, we use $z_k=-ik$ as defined in Eq.~(\ref{defzzk}):
\begin{eqnarray}
\phi_{k, b_1} &=& \pi + c_{(0)} + c_{(1)} z_k + c_{(2)} z_k^2 + ...\\
\phi_{k, b_2} &=& \pi + d_{(0)} + d_{(1)} z_k + d_{(2)} z_k^2 + ...
\end{eqnarray}

Plugging these into the expression of Chern number Eq.~(\ref{ck0}), we have:
\begin{eqnarray}
\mC &=& \frac{1}{2\pi i}\left(c_{(1)} b_2 - d_{(1)} b_1\right)\label{shiftBZ}\\
&+& \frac{1}{2\pi}\bigg[\left(d_{(2)}b_1^2-c_{(2)}b_2^2\right)+2iz_{k_0}\left(d_{(2)}b_1-c_{(2)}b_2\right)\bigg] + ...\nonumber
\end{eqnarray}
where $z_{k_0} = -ik_0 = -i(\bm k_{0,x}+i\bm k_{0,y})/\sqrt{2}$.

Since the origin of the BZ $k_0$ is an arbitrary choice, the Chern number must be independent on $k_0$. This must be true for $b_{1,2}$ of generic geometry (including aspect ratio and torus angle). Therefore Eq.~(\ref{shiftBZ}) leads to the conditions that (i) all terms except the first line should vanishes identically for generic $b_{1,2}$, and (ii) the first line should equal to the Chern number $\mC=1$. Condition (i) implies $c_{(n)>1}=d_{(n)>1}=0$. The most general solution for (ii) is:
\begin{equation}
c_{(1)} = b_{1}^* - \kappa\cdot b_1,\quad d_{(1)} = b_{2}^* - \kappa\cdot b_2,\label{c1d1}
\end{equation}
where the parameter $\kappa$ is a gauge choice: we first set $\kappa=0$ and derive the $\sigma$ function wavefunction as a general solution; then nonzero $\kappa$ can be generated by the gauge transformation discussed in the first section.

The Chern number constrain (ii) is satisfied by Eq.~(\ref{c1d1}) by noticing that the fact that primitive reciprocal lattice vectors $\bm b_{1,2}$ span an area $S=1$ according to Eq.~(\ref{area}) and Eq.~(\ref{bvector}) (we have set $S=1$ as the unit of area). Therefore we have:
\begin{equation}
b_1^*b_2-b_1b_2^* = 2\pi i.
\end{equation}

We now try to fix the constant terms by choosing the origin of $\bm k-$space. We start by writing the four boundary conditions as follows:
\begin{eqnarray}
\phi_{k,b_1} &=& \pi + c_{(0)} + b^*_1\cdot z_k,\nonumber\\
\phi_{k,b_2} &=& \pi + d_{(0)} + b^*_2\cdot z_k,\nonumber\\
\phi_{k,-b_1} &=& \pi + c'_{(0)} - b^*_1\cdot z_k,\nonumber\\
\phi_{k,-b_2} &=& \pi + d'_{(0)} - b^*_2\cdot z_k.\nonumber
\end{eqnarray}
where the constant terms are constrained by Eq.~(\ref{phiinvertk}) as follows (modulo $2\pi$):
\begin{equation}
c_{(0)}+c'_{(0)} = |b_1|^2/i,\quad d_{(0)}+d'_{(0)} = |b_2|^2/i.
\end{equation}

Shifting the origin of the $\bm k-$space coordinate system tunes the differences $c_{0}-c'_{(0)}$ and $d_{0}-d'_{(0)}$. We can use this degrees of freedom to get $c_{(0)}=c'_{(0)}$ and $d_{(0)}=d'_{(0)}$. As a result, we get the boundary condition shown in Eq.~(\ref{complexphi}). From now on, we set the origin of the BZ at the point where the boundary condition Eq.~(\ref{complexphi}) is satisfied.
\end{proof}

Since the wavefunction $\tilde{u}_k$ is holomorphic in $k$, it is uniquely determined by its $\bm k-$space boundary condition \cite{haldanetorus1}. Note that we have shown the LLL wavefunction satisfies the $\bm k-$space boundary Eq.~(\ref{complexphi}) as calculated in Eq.~(\ref{kbc_u}). This leads to the conclusion that the ideal flatband with $\mC=1$ has the LLL wavefunction character and more explicitly its wavefunction is:
\begin{eqnarray}
\psi_{\bm k}(\bm r) &=& N_{\bm k}\tilde{u}_k(\bm r)\exp(i\bm k\cdot\bm r),\nonumber\\
&=& \mathcal{N}_{\bm k}\mathcal{B}(\bm r)\Phi_{\bm k}(\bm r),\label{ifbwf}
\end{eqnarray}
where the $\bm k-$independent function $\mathcal{B}(\bm r)$ is model-dependent, but the LLL character is universal. Note that the LLL wavefunction $\Phi_{\bm k}(\bm r)$ has a factor $\exp(-|k|^2/2)$. We hence have the two normalization factors $N_{\bm k}$ and $\mathcal{N}_{\bm k}$ related by:
\begin{equation}
N_{\bm k} = \mathcal{N}_{\bm k}\exp(-|k|^2/2),\label{twonorm}
\end{equation}
where $\mathcal{N}_{\bm k}$ is periodic: $\mathcal{N}_{\bm k}=\mathcal{N}_{\bm k+\bm b}$.

Performing a gauge transformation as discussed in the previous section generates nonzero $\kappa$ in Eq.~(\ref{c1d1}), which gives wavefunction in other representation: for instance, the Jacobi theta function representation is derived in Landau gauge, which has a periodic boundary condition in one direction ($c_{(1)}=0$ or $d_{(1)}=0$) and a quasi-periodic boundary condition in the other. The proof of the uniqueness of the $\mC=1$ ideal flatband wavefunction essentially follows from the fact that the possible forms of quasi-periodic holomorphic functions with a single zero in the periodic domain is highly constrained.

\subsection{Normalization factor and Kahler potential}
In this section, we discuss the band geometry of the wavefunction Eq.~(\ref{ifbwf}). We show the $\bm k-$space normalization factor $\mathcal{N}_{\bm k}$ controls the fluctuation of band geometry, and derive its explicit expression.

We define the holomorphic and anti-holomorphic Berry connection as,
\begin{eqnarray}
A_k &\equiv& \omega_a^*\bm A^a_{\bm k} = i\langle u_{\bm k}|\partial_ku_{\bm k}\rangle,\nonumber\\
\bar{A}_k &\equiv& \omega_a\bm A^a_{\bm k} = i\langle u_{\bm k}|\bar{\partial}_ku_{\bm k}\rangle.
\end{eqnarray}

By using the holomorphic property of the cell-periodic wavefunction, we get,
\begin{eqnarray}
\bar{A}_k &=& i\int_{\bm r} N_{\bm k}\tilde{u}^*_{k}(\bm r) \bar{\partial}_k\left(N_{\bm k}\tilde{u}_{k}(\bm r)\right),\nonumber\\
&=& i\left(N^{-1}_{\bm k}\bar{\partial}_kN_{\bm k}\right)\int_{\bm r} N^2_{\bm k} \tilde{u}^*_{k}(\bm r)\tilde{u}_k(\bm r),\nonumber\\
&=& i\bar{\partial}_k\log N_{\bm k}.
\end{eqnarray}

Using Eq.~(\ref{twonorm}), we get the holomorphic and anti-holomorphic parts of the Berry connections as:
\begin{eqnarray}
\bar{A}_k &=& -\frac{i}{2}k + i\bar{\partial}_k\log\mathcal{N}_{\bm k} = \left(A_k\right)^*.\label{holomorphicA}
\end{eqnarray}

In holomorphic coordinates, the Berry curvature $\Omega_{\bm k} \equiv \epsilon_{ab}\partial^a_{\bm k}\bm A^a_{\bm k}$ is expressed as,
\begin{eqnarray}
\Omega_{\bm k} &=& -i\left(\partial_k\bar{A}_k - \bar{\partial}_kA_k\right),\nonumber\\
&=& -1 + \left(\partial_k\bar{\partial}_k+\bar{\partial}_k\partial_k\right)\log\mathcal{N}_{\bm k},\nonumber\\
&=& -1 + \tilde{g}_{ab}\partial_{\bm k}^a\partial_{\bm k}^b\log\mathcal{N}_{\bm k},\label{OmegaN}
\end{eqnarray}
where $\tilde{g}_{ab} = \omega_a^*\omega_b + \omega_a\omega_b^*$ is the unimodular part of the Fubini-Study metric. This shows that the logarithm of the normalization factor is the $\bm k-$space K\"ahler potential that controls the fluctuation of Berry curvature, as discussed in the main text.

We now derive explicit expressions for band geometry. Following the notations used in the main text, we wrote the wavefunction and its Fourier modes are:
\begin{equation}
\psi_{\bm k}(\bm r) = \mathcal{N}_{\bm k}\mathcal{B}(\bm r)\Phi_{\bm k}(\bm r),\quad |\mathcal{B}(\bm r)|^2 = \sum_{\bm b}w_{\bm b}e^{i\bm b\cdot\bm r}.
\end{equation}
where $\Phi_{\bm k}(\bm r)$ is the quantum Hall wavefunction. The form factor determines the overlap of two wavefunction in $\bm k-$space. Using magnetic translation algebra, they are calculated as follows:
\begin{eqnarray}
&&\langle u_{\bm k_1}|u_{\bm k_2+\bm b}\rangle = \int_{\bm r} e^{i(\bm k_1-\bm k_2-\bm b)\cdot\bm r}\psi^*_{\bm k_1}(\bm r)\psi_{\bm k_2+\bm b}(\bm r),\label{holomorphicbandformfactor}\\
% &=& \mathcal{N}_{\bm k_1,\bm k_2+\bm b}\int_{\bm r} e^{i(\bm k_1-\bm k_2-\bm b)\cdot\bm r} |\mathcal{B}(\bm r)|^2 \Phi^*_{\bm k_1}(\bm r)\Phi_{\bm k_2+\bm b}(\bm r),\nonumber\\
% &=& \mathcal{N}_{\bm k_1,\bm k_2+\bm b}\int_{\bm r} \left(\sum_{\bm b'}w_{\bm b'} e^{i(\bm k_1-\bm k_2-\bm b+\bm b')\cdot\bm r}\right)\Phi^*_{\bm k_1}\Phi_{\bm k_2+\bm b}(\bm r),\nonumber\\
&=& \mathcal{N}_{\bm k_1,\bm k_2+\bm b}\sum_{\bm b'}w_{\bm b'}\langle\bm k_1|e^{i(\bm k_1-\bm k_2-\bm b+\bm b')(\bm R+\bm{\bar R})}|\bm k_2+\bm b\rangle,\nonumber
\end{eqnarray}
where $\int_{\bm r}\equiv\int d^2\bm r$ and $\mathcal{N}_{\bm k_1,\bm k_2}\equiv\mathcal{N}_{\bm k_1}\mathcal{N}_{\bm k_2}$. The $|\bm k\rangle$ represents a LLL state. The term $\langle\bm k_1|e^{i\bm q\cdot(\bm R+\bm{\bar R})}|\bm k_2+\bm b\rangle$ appearing in the last line of the above equation is just the form factor of the quantum Hall wavefunction. We denote it as $f^{\bm k\bm k'}_{\bm b}$, and its expression can be explicitly calculated based on the magnetic translation algebra:
\begin{eqnarray}
f^{\bm k\bm k'}_{\bm b} &\equiv& \langle\bm k|e^{i(\bm k-\bm k'-\bm b)\cdot(\bm R+\bar{\bm R})}|\bm k'\rangle,\label{symgauformfactor}\\
&=& \eta_{\bm b}e^{\frac{i}{2}(\bm k+\bm k')\times\bm b}e^{\frac{i}{2}\bm k\times\bm k'}e^{-\frac14|\bm k-\bm k'-\bm b|^2}.\nonumber
\end{eqnarray}

% With Eq.~(\ref{holomorphicbandformfactor}) and Eq.~(\ref{symgauformfactor}), we get the form factor of the $\mC=1$ ideal flatbands:
% \begin{equation}
% \langle u_{\bm k_1}|u_{\bm k_2+\bm b}\rangle = \mathcal{N}_{\bm k_1}\mathcal{N}_{\bm k_2+\bm b}\sum_{\bm b'}w_{\bm b'}f^{\bm k_1,\bm k_2+\bm b}_{-\bm b'}.
% \end{equation}

Normalization factors are computed straightforwardly by taking $\bm k_1=\bm k_2$, $\bm b=\bm 0$:
\begin{equation}
\mathcal{N}^{-2}_{\bm k} = \sum_{\bm b'}w_{\bm b'}f^{\bm k,\bm k}_{-\bm b'} = \sum_{\bm b'}\eta_{\bm b'}w_{\bm b'}e^{i\bm k\times\bm b'}e^{-\frac{1}{4}\bm b'^2},\label{normfactoranalytic}
\end{equation}
which determines the band geometry through Eq.~(\ref{OmegaN}).

\section{Interacting Hamiltonian and Numerics}
In this section, we work with the model wavefunction derived in the main text and provide details of the computation of band geometry and interacting Hamiltonians.

\subsection{Interacting Hamiltonians}
\subsubsection{Interaction in a Landau level}
We first review the interaction Hamiltonian for translational invariant two-body interaction in a Landau level. With $v_{\bm q}$ as the Fourier transform of the interaction, the Hamiltonian is,
\begin{equation}
H = \sum_{\bm q}v_{\bm q}e^{i\bm q\cdot(\bm r_1-\bm r_2)} = \sum_{\bm q}v_{\bm q}e^{i\bm q\cdot(\bm R_1-\bm R_2)}e^{i\bm q\cdot(\bar{\bm R}_1-\bar{\bm R}_2)},\label{QHint}
\end{equation}
where $\bm R$ and $\bar{\bm R}$ are respectively the guiding centers and Landau orbits introduced earlier. The matrix element is,
\begin{eqnarray}
H_{\bm k_1\bm k_2;\bm k_3\bm k_4} &=& \sum_{\bm q}v_{\bm q}\langle\bm k_1|e^{i\bm q\cdot(\bm R+\bar{\bm R})}|\bm k_4\rangle\langle\bm k_2|e^{-i\bm q\cdot(\bm R+\bar{\bm R})}|\bm k_3\rangle\nonumber\\
&\times& \langle\bm k_2|e^{-i(\bm k_1-\bm k_4-\bm b)\cdot(\bm R+\bar{\bm R})}|\bm k_3\rangle,\nonumber\\
&=& \sum_{\bm b}v_{\bm k_1-\bm k_4-\bm b}f^{\bm k_1\bm k_4}_{\bm b}f^{\bm k_2\bm k_3}_{-\bm b+\delta\bm b},\label{H1234}
\end{eqnarray}
where we substituted $\bm q$ with $\bm k_1-\bm k_4-\bm b$ and replaced the summation over $\bm q$ with the summation over $\bm b$. We also introduced a new reciprocal lattice vector:
\begin{equation}
\delta\bm b=\bm k_1+\bm k_2-\bm k_3-\bm k_4.\label{def_deltab}
\end{equation}

The second quantized form for numerical exact diagonalization calculation is:
\begin{equation}
\hat{H} = H_{\bm k_1\bm k_2;\bm k_3\bm k_4}c^{\dag}_{\bm k_1}c^{\dag}_{\bm k_2}c_{\bm k_3}c_{\bm k_4},\label{secondquantizedform}
\end{equation}
where $c^{\dag}_{\bm k}$ creates a LLL electron with magnetic translation quantum number $\bm k$.

\subsubsection{Interaction in C=1 ideal flatband}
The interacting matrix element for generic $\mC=1$ ideal flatbands are derived from the wavefunction $\mathcal{N}_{\bm k}\mathcal{B}(\bm r)\Phi_{\bm k}(\bm r)$. We derive the effective interaction in quantum Hall basis:
\begin{eqnarray}
&&\left(\prod_{i=1}^4\mathcal{N}_{\bm k_i}\right)|\mathcal{B}(\bm r_1)|^2v(\bm r_1-\bm r_2)|\mathcal{B}(\bm r_2)|^2,\label{FQHumklapp}\\
&=& \left(\prod_{i=1}^4\mathcal{N}_{\bm k_i}\right)\sum_{\bm q,\bm b_{i,j}}v(\bm q)\left(w_{\bm b_i}w_{\bm b_j}\right)e^{i(\bm b_i+\bm q)\cdot\bm r_1}e^{i(\bm b_j-\bm q)\cdot\bm r_2}.\nonumber
\end{eqnarray}
where on the right-hand-side, we have substituted the Fourier components,
\begin{equation}
|\mathcal{B}(\bm r)|^2 = \sum_{\bm b}w_{\bm b}e^{i\bm b\cdot\bm r},\quad v(\bm r) = \sum_{\bm q}v(\bm q)e^{i\bm q\cdot\bm r}.
\end{equation}

Comparing the translational invariant interaction matrix element in Eq.~(\ref{QHint}) and Eq.~(\ref{H1234}), we notice that the only difference is to shift the subscript of $f_{\bm q}^{\bm k\bm k'}$ by $-\bm b_{i,j}$. Hence, we arrive at,
\begin{equation}
H_{\bm k_1\bm k_2;\bm k_3\bm k_4}=\left(\prod_{i=1}^4\mathcal{N}_{\bm k_i}\right)h_{\{\bm k\}},\label{MatrixEle0}
\end{equation}
where
\begin{eqnarray}
&&h_{\bm k_1\bm k_2;\bm k_3\bm k_4} = \label{MatrixEle}\\
&&v_{\bm k_1-\bm k_4-\bm b}\left(\sum_{\bm b_i}w_{\bm b_i}f^{\bm k_1,\bm k_4}_{\bm b-\bm b_i}\right)\left(\sum_{\bm b_j}w_{\bm b_j}f^{\bm k_2,\bm k_3}_{-\bm b+\delta\bm b-\bm b_j}\right).\nonumber
\end{eqnarray}

Eq.~(\ref{MatrixEle0}), Eq.~(\ref{MatrixEle}) and Eq.~(\ref{normfactoranalytic}) form a numerical \emph{exact} description of the interacting problems in $\mC=1$ ideal bands in the LLL basis, provided that all Fourier modes $w_{\bm b}$ are known.

\subsubsection{Example: chiral twisted bilayer graphene}
The single-particle Hamiltonian for twisted bilayer graphene consists of a standard single-layer graphene Hamiltonian for the top/bottom layer and an interlayer coupling whose periodicity defines the moir\'e superlattice. Following Bistritzer and MacDonald \cite{Bistritzer12233}, the effective continuum Hamiltonian of a single valley is,
\begin{equation}
H_{BM}=\int d^2\bm r \Psi_{BM}^\dagger(\bm r)\left(\begin{array}{cc}h_D^{b}\left(\frac{\theta}{2}\right)& T(\bm r) \\T^\dagger (\bm r)& h_D^{t}\left(-\frac{\theta}{2}\right)\end{array}\right)\Psi_{BM}(\bm r).\label{HTBG}
\end{equation}

A related Hamiltonian can be found for the opposite valley by acting with time reversal symmetry. The continuum approximation to the Dirac Hamiltonian of a layer $\lambda=t,b$ is:
\begin{equation}
h_D^{\lambda}\left(\frac{\theta}{2}\right)=\frac{\sqrt{3}at_0}{2}\left(-i{\bm \nabla}-{\bm K}_+^{\lambda}\right)\cdot e^{-\frac{i\theta}{4}\sigma_z}\bm{\sigma}e^{\frac{i\theta}{4}\sigma_z}.\label{hktheta}
\end{equation}
where $\bm K_+^{t/b}$ is the graphene Dirac point $\bm K_+$ rotated by $\pm\theta/2$. We define the moir\'e Dirac points as $\bm K$ = $\bm K_+^b-\bm K_+^{\Gamma}$, $\bm K'$ = $\bm K_+^t-\bm K_+^{\Gamma}$ where $\bm K_+^{\Gamma}$ is the moir\'e Gamma point labeled in graphene's reciprocal lattice coordinates. The interlayer tunneling potential $T(\bm{r})$ is constrained by the symmetries of a single valley: $\mC_3$, $\mathcal{M}_y$ and $\mC_2\mathcal{T}$:
\begin{equation}
T(\bm r)=\sum_{j=0}^2T_je^{-i(\bm q_0-\bm q_j)\cdot\bm r}.\label{T}
\end{equation}
with $\phi$=$2\pi/3$, the $T_j$ is:
\begin{equation}
T_j= \omega_0 - \omega_1\cos(j\phi)\sigma_x + \omega_1\sin(j\phi)\sigma_y.\label{Tjdef}
\end{equation}

The chiral limit is obtained by setting $w_0=0$, where at magic twisted angles the energy dispersion at charge neutrality becomes exactly flat \cite{Grisha_TBG}. We took other parameters (moir\'e interlayer inter-sublattice hopping parameter $w_1$, graphene's hopping parameter $t_0$, graphene's lattice constant $a$) of the continuum model as follows:
\begin{eqnarray}
w_1 = 110\text{meV},\quad t_0 = 2.62\text{eV},\quad a = 2.46\text{\AA},
\end{eqnarray}
and the first magic angle is $\theta = 1.132^{\circ}$ where the chiral model has exactly flat bands.

The wavefunction of chiral TBG has been found to have a LLL character in Ref.~(\onlinecite{JieWang_NodalStructure}):
\begin{eqnarray}
\psi_{\bm k}(\bm r)=\mathcal{N}_{\bm k}\left(\begin{matrix}i\mathcal{G}(\bm r)\\\eta\mathcal{G}(-\bm r)\end{matrix}\right)\Phi_{\bm k}(\bm r),\label{cTBGwf}
\end{eqnarray}
where $\eta$ is the intra-valley inversion eigenvalue, and $\eta=+1$ at the first magic angle \cite{JieWang_NodalStructure}. Eq.~(\ref{cTBGwf}) is in the same form as our $\mC=1$ ideal flatband wavefunction, with $\mathcal{B}(\bm r)$ replaced by a two-component layer spinor. We then Fourier transform $|\mG(\bm r)|^2=\sum_{\bm b}u_{\bm b}\exp(i\bm b\cdot\bm r)$:
\begin{equation}
|\mG(\bm r)|^2 = u_{\bm 0} + \sum_{i=1,2,3}\left(u_{\bm b_1}e^{i\bm b_i\cdot\bm r} + h.c.\right) + ...
\end{equation}
where $\bm b_{1,2,3}$ are illustrated in the main text. Here we choose $\Phi_{\bm k}(\bm r)$ in the symmetric gauge represented by the Weierstrass sigma function. The Fourier modes are numerically computed:
\begin{equation}
u_{\bm 0} = 0.445,\quad u_{\bm b_{1,2,3}} = 0.108 - 0.035\mathrm{i}. \label{u0u1value}
\end{equation}

For layer isotropic interaction, only $|\mG(\bm r)|^2+|\mG(-\bm r)|^2$ enters in the interaction Hamiltonian. We denote its Fourier mode as $w_{\bm b}$:
\begin{equation}
|\mG(\bm r)|^2 + |\mG(-\bm r)|^2 = \sum_{\bm b}w_{\bm b}\exp(i\bm b\cdot\bm r),
\end{equation}
and we have:
\begin{equation}
w_{\bm b} = 2\Re{u}_{\bm b},
\end{equation}
for all reciprocal lattice vectors $\bm b$. Plots of $w_{\bm b}$ are shown in the main text. Results derived in the previous sections directly applies to the cTBG flatbands.

\subsubsection{Quantum Hall model with Umklapp interactions}
The model Eq.~(\ref{FQHumklapp}) describes exactly a quantum Hall model with Umklapp interactions. It turns out that this model has exact zero modes and interesting phase transitions as discussed in the main text. The exact model Eq.~(\ref{FQHumklapp}) can be approximated well by retaining few parameters.

Note that the LLL projection replaces any $e^{i\bm q\cdot\bm r}$ to:
\begin{equation}
e^{i\bm q\cdot\bm r} \rightarrow e^{-\frac14\bm q^2}e^{i\bm q\cdot\bm R},
\end{equation}
which means that the Umklapp interactions with large momentum transfer are less important as they are suppressed by the Gaussian form factor. The minimal Hamiltonian retaining the leading order Umklapp terms (associated to shortest primitive vectors $\bm b_{1,2}$) are:
\begin{equation}
H = \sum_{\bm q}\left(\tilde{w}_0 + \sum_{i=1}^n\sum_{j=1}^2(\tilde{w}_1e^{i\bm b_i\cdot\bm r_j} + h.c.)\right)v_{\bm q}e^{i\bm q\cdot(\bm r_1-\bm r_2)}.\label{simplemodel}
\end{equation}
where $j=1,2$ labels particle, and $i$ runs from $1$ to $n$ labeling the primitive reciprocal lattice basis.

On square torus, the primitive lattice vectors are $\bm b_{1,2}$ with $|\bm b_1|=|\bm b_2|$ and $\bm b_1\cdot\bm b_2=0$, so we have $n=2$. The $C_4$ symmetry requires $\tilde{w}_1$ to be a real number. On triangular torus, we have $i$ runs from $1$ to $3$ labeling the two reciprocal lattice vectors $\bm b_{1,2}$ and the third one that is related $\bm b_{1,2}$ by the three-fold rotation symmetry: $\bm b_3=-(\bm b_1+\bm b_2)$. The primitive vectors are plotted in the main text.

The Umklapp interaction parameters $\tilde{w}_n$ are easily obtained from the band geometry $w_{\bm b}$ and is lattice-geometry dependent. For instance, on $\mC_4$ symmetric square lattice, starting from expanding Eq.~(\ref{FQHumklapp}):
\begin{eqnarray}
H &=& \sum_{\bm q}e^{i\bm q\cdot(\bm r_1-\bm r_2)}v(\bm q)\times\bigg[w_{\bm 0}^2 + \label{1eq}\\
&& w_{\bm 0} \left(w_{\bm b_1}e^{i\bm b_1\cdot\bm r_1} + w_{\bm b_1}e^{i\bm b_2\cdot\bm r_1} + \bm r_1\leftrightarrow\bm r_2 + h.c.\right)\nonumber\\
&+& |w_{\bm b_1}|^2\left(e^{i\bm b_1\cdot(\bm r_1-\bm r_2)} + e^{i\bm b_2\cdot(\bm r_1-\bm r_2)} + h.c.\right) + ...\bigg],\nonumber
\end{eqnarray}
where $(...)$ contains terms of $|\bm b|\geq|\bm b_{1,2}|$ which we neglect. We can then absorb the last line into the first line by shifting $\bm q$ (it is summed eventually). For the $v_1$ pseudopotential $v(\bm q)=\bm q^2$, and momentum shifting yields $v(\bm q)\rightarrow(\bm q\pm\bm b_{1,2})^2 = \bm q^2 \pm\bm b_{1,2}\cdot\bm q+\bm b^2_{1,2}$ where the last two terms are the $v_0$ anisotropic and isotropic pseudopotentials which do not couple to fermions. In general, such shift will not affect pseudopotential of higher orders. We thus arrive at:
\begin{equation}
w^2_{\bm 0}v(\bm q)e^{i\bm q\cdot(\bm r_1-\bm r_2)}\bigg[1 + 4|\frac{w_{\bm b_1}}{w_{\bm 0}}|^2 + \sum_{i,j=1}^2\left(\frac{w_{\bm b_1}}{w_{\bm 0}}e^{i\bm b_i\cdot\bm r_j} + h.c.\right)\bigg],\nonumber
\end{equation}
from which we see,
\begin{equation}
\tilde{w}_0 = w_{\bm 0}^2 + 4|w_{\bm b_1}|^2,\quad\tilde{w}_1 = w_{\bm b_1}w_{\bm 0}.
\end{equation}

A similar calculation is easily generalized to the triangular lattice (as in cTBG) where,
\begin{equation}
\tilde{w}_0 = w_{\bm 0}^2 + 6|w_{\bm b_1}|^2,\quad\tilde{w}_1 = w_{\bm b_1}w_{\bm 0} + w^{*2}_{\bm b_1}.
\end{equation}

Taking into account the normalization factors, the interacting Hamiltonian is invariant under rescaling $w_{\bm b}\rightarrow \alpha w_{\bm b}$ for all $\bm b$. In the main text, we have set $w_{\bm 0}=1$.

\subsection{COM Pseudopotentials without rotational symmetry}
A general two-body interaction projected to a single LL is given by,
\begin{equation}
H = \int\frac{d^2\bm q_1}{(2\pi)^2}\int\frac{d^2\bm q_2}{(2\pi)^2}V_{\bm q_1,\bm q_2}\rho_{\bm q_1}\rho_{\bm q_2},\label{appendgeneralH}
\end{equation}
where
\begin{equation}
\rho_{\bm q} = \sum_i\exp(i\bm q\cdot\bm R_i),
\end{equation}
is the LLL projected guiding center density operator. Eq.~(\ref{appendgeneralH}) includes Eq.~(\ref{FQHumklapp}), which we discussed in the previous section.

To begin with, we review the simplest pseudopotential formalism with both translational and rotational symmetry. In this case $V_{\bm q_1,\bm q_2}=V_{|\bm q|}\delta_{\bm q_1,\bm q}\delta_{\bm q_2,-\bm q}$ and the Hamiltonian is block-diagonal in the angular momentum basis $m$:
\begin{equation}
\langle m|H|m'\rangle = c_m\delta_{m,m'},
\end{equation}
where $c_m$ is the pseudopotentials that is related to the interaction $V_{|\bm q|}$ by:
\begin{equation}
V_{|\bm q|} = \sum_{m=0}^{\infty}c_mV_m(|\bm q|),\quad V_m(|\bm q|) = e^{-\frac{1}{2}\bm q^2}L_m(\bm q^2),
\end{equation}

The Hamiltonian can then be rewritten as a sum of projectors:
\begin{equation}
H = \sum_{i<j}\sum_mc_mP_m(\bm R_i-\bm R_j),
\end{equation}
where the projectors are:
\begin{equation}
P_m(\bm R_i-\bm R_j)\equiv 2\int\frac{d^2\bm q}{(2\pi)^2}L_m(\bm q^2)e^{-\frac12\bm q^2}e^{i\bm q\cdot(\bm R_i-\bm R_j)}.
\end{equation}
satisfying $P_m(\bm R_i-\bm R_j)P_n(\bm R_i-\bm R_j)=\delta_{mn}P_m(\bm R_i-\bm R_j)$, which projects a pair of electron into their relative angular momentum $m$ sector. For this reason, the Laughlin $\nu=1/3$ state is the exact zero-energy eigenstate of $c_{m}\propto\delta_{m,1}$ interactions because the Laughlin wavefunction has zero weight in the $m=1$ channel for any electron pair.

With translation symmetry but without rotation symmetry, we have $V_{\bm q_1,\bm q_2}=V_{\bm q_1}\delta_{\bm q_1,-\bm q_2}$. The projector Hamiltonian can still be constructed by using the generalized-Laguerre polynomials as shown in Ref.~(\onlinecite{generalizedPP_PRL}). The generalized pseudopotentials are
\begin{eqnarray}
V_{\bm q} &=& \sum_{m,n=0,\sigma=\pm}^{\infty}c^{\sigma}_{m,n}V^{\sigma}_{m,n}(\bm q),\nonumber\\
c^{\sigma}_{m,n} &=& \int d^2\bm qV_{\bm q}V^{\sigma}_{m,n}(\bm q),
\end{eqnarray}
where
\begin{eqnarray}
V^{+}_{m,n}(\bm q) &=& \lambda_n\mathcal{N}_{mn}\left(L^n_m(\bm q^2)e^{-\frac12\bm q^2}\bm q^n + c.c\right),\nonumber\\
V^{-}_{m,n}(\bm q) &=& -i\mathcal{N}_{mn}\left(L^n_m(\bm q^2)e^{-\frac12\bm q^2}\bm q^n - c.c\right).
\end{eqnarray}
where the normalization factors are $\mathcal{N}_{mn}=\sqrt{2^{n-1}m!/\left(\pi(m+n)!\right)}$, and $\lambda_n=1/\sqrt2$ for $n=0$ or $\lambda_n=1$ for $n\neq0$.

The Hamiltonian can be rewritten as,
\begin{equation}
H=\sum_{i<j}\sum_{m,n,\sigma}c^{\sigma}_{mn}P^{\sigma}_{mn}(\bm R_i-\bm R_j),
\end{equation}
where
\begin{equation}
P^{\sigma}_{mn}(\bm R_i-\bm R_j)=\int\frac{d^2\bm q}{(2\pi)^2}V^{\sigma}_{mn}(\bm q)e^{i\bm q\cdot(\bm R_i-\bm R_j)},
\end{equation}
are the generalized projector.

In this work, we generalize the pseudopotential formalism to the most general form with neither translation nor rotation symmetry. It is then necessary to consider center-of-mass dependent interaction as discussed in the main text. We defined center-of-mass guiding center coordinates $\bm R^{+}_{ij}\equiv(\bm R_i+\bm R_j)/\sqrt2$ which commutates with the relative coordinates $\bm R^{-}_{ij}\equiv(\bm R_i-\bm R_j)/\sqrt2$: they separately form independent algebras and yields independent pseudopotential projectors. To see this more precisely, we rewrite the interaction as:
\begin{equation}
V_{\bm q_1,\bm q_2} = \sum_{\sigma=\pm;i\neq j}\exp\left(i\bm Q^{\sigma}\cdot\bm R^{\sigma}_{ij}\right),
\end{equation}
with $\bm Q^{\pm}=(\bm q_1\pm\bm q_2)/\sqrt2$. The Hamiltonian can be written as:
\begin{equation}
H = \sum_{i<j}\sum_{mn\sigma}\sum_{m'n'\sigma'}c^{\sigma\sigma'}_{mm';nn'}P^{\sigma}_{mn}(\bm R_i-\bm R_j)P^{\sigma'}_{m'n'}(\bm R_i+\bm R_j).\label{COMPP}
\end{equation}

The translation non-symmetric but rotational symmetric case considered in the main text is a simplified version of Eq.~(\ref{COMPP}), with $m=n$ as the relative angular momentum and $m'=n'$ as the COM angular momentum.

The COM pseudopotentials in principle can be read off from the energy spectrum of interacting two particles \cite{zhao_nonabelian,hierarchy_FCI}, and will be useful for numerical exploration. We leave this for future work.

\subsection{Exact many-body zero modes}
In the main text, we have shown that the existence of the exact many-body zero modes essentially follows from their independence on the COM pseudopotential projector. In this section of SM, following the same idea, we provide an alternative proof based on the clustering properties of wavefunctions.

\begin{theorem}
The two-body interacting Hamiltonian $\hat{P}H\hat{P}$ processes three fold exact zeros modes (not necessarily lowest eigenvalue) at filling fraction $\nu=N/N_{\phi}=1/3$, with the $v_1$ relative Haldane pseudopotential $v_{\bm q}=\bm q^2$:
\begin{equation}
H = \sum_{\bm q;m,n;i,j}v_{\bm q}\tilde{w}_{mn}\left(e^{i(\bm b_m+\bm q)\cdot\bm r_i+i(\bm b_n-\bm q)\cdot\bm r_j}+i\leftrightarrow j\right)+H.c.,
\end{equation}
where $\hat{P}$ is the lowest Landau level projection operator:
\begin{equation}
\hat{P}\bm r\hat{P} = \bm R,
\end{equation}
and,
\begin{equation}
\bm b_n = n_1\bm b_1+n_2\bm b_2,\quad n_{1,2}\in\mathbb{Z}.
\end{equation}

The many-body zero-modes above are the $\nu=1/3$ Laughlin states. Generalization to other fractional quantum Hall states is a straightforward step.
\end{theorem}
\begin{proof}
We define a symmetric function,
\begin{equation}
\mathcal{F}(\bm r_1, \bm r_2) \equiv \sum_{m,n}\tilde{w}_{mn}\left(e^{i\bm b_m\cdot\bm r_1+i\bm b_n\cdot\bm r_2}+1\leftrightarrow2\right) + H.c.
\end{equation}

The effective lattice translational symmetric interaction is:
\begin{equation}
\tilde{v}(\bm r_1, \bm r_2) = \mathcal{F}(\bm r_1, \bm r_2)v(\bm r_1-\bm r_2).
\end{equation}

For the $v_1$ Haldane pseudopotential, we have $v(\bm r_1-\bm r_2) = \delta^{''}(\bm r_1-\bm r_2)$, and we see that the short ranged property of the interaction $\tilde{v}$ is still fully determined by the relative part $v$. From this we know that Laughlin $\nu=1/3$ wavefunction, in which the decay power is strictly faster than first order for any pair of electrons coinciding at the same location, must still be the exact zero energy eigenstate for $\tilde{v}$.
\end{proof}

Including the normalization factors does not affect the existence of zero modes at all. The problems with normalization factors are formalized in terms of generalized eigenvalue problems with non-orthonormal basis:
\begin{equation}
H_{mn}\tilde{\psi}_n = \lambda S_{mn}\tilde{\psi}_n,\label{geneigenprob}
\end{equation}
where $S_{mn}\equiv\langle\tilde{\psi}_m|\tilde{\psi}_n\rangle$ is the overlap matrix. In our problem, $\tilde{\psi}$ represents a LLL many-body unnormalized determinant wavefunction, and $H$ is the interacting Hamiltonian $\tilde{v}(\bm r_1,\bm r_2)$, and $S$ is diagonal. The generalized eigenvalue problem are solved by two steps: (i) find $D$ {\it s.t.} $D^{-1}SD$ is diagonal with diagonal elements $\sigma$, and construct $A_{ij}\equiv D_{ij}/\sqrt{\sigma_j}$, (ii) the eigenvalue problem is transformed into the ordinary eigenvalue problem for $A^{\dag}HA$. The existence of zero modes are unaffected by the overlap matrix $S$: if $H\Psi=0$, then $A^{-1}\Psi$ is the zero modes of the generalized problem Eq.~(\ref{geneigenprob}).

\section{Further Applications}
In addition to the topological phase transition discussed in the main text, in this section we discuss two more applications of the ideal flatband theory.

\subsection{Effective interaction in $C_2T$ symmetric TBG and impact on superconductivity}

We have discussed how band geometry influences interacting physics in a single topological flatband. Particularly, we discussed the interaction in TBG with an hBN substrate (which breaks $C_2T$ symmetry) which polarizes a single TBG flatband. This is a relevant experimental set up where anomalous Hall effects were observed \cite{Sharpe605,Serlin900}. Superconductivity, on the other hand, was observed in TBG devices with $C_2T$ symmetry where all four flatbands (ignoring spin) are equally important \cite{Cao:2018aa,Cao:2018ab}. In this section, we generalize our discussion to the time-reversal invariant TBG flatbands, and comment on how inhomogeneous band geometry would affect superconductivity.

The four flatbands are labeled by two numbers $\alpha=(v,c)$, {\it i.e.} valley $v$ and Chern number $c$. The four bands can be grouped into two pairs distinguished by their Chern number. For this reason, each group is termed a ``Chern sector'', consisting of two flatbands of the same Chern number but opposite valley. If regarding valley as pseudo-spin, this is a situation very similar to two quantum Hall ferromagnets related by time-reversal symmetry.

The low-lying excitations in each Chern sector are charged skyrmions. Ref.~(\onlinecite{chargedskyrmion}) proposed a superconductivity mechanism based on the skyrmion pairing. Importantly, the perturbative calculations of Ref.~(\onlinecite{chargedskyrmion}) indicated that the skyrmion coupling is anti-ferromagnetic which consequently favors a skyrmion bound state.

Subsequently, Ref.~(\onlinecite{skyrmionSC_DMRG}) numerically studied the possibility of skyrmion superconductivity in a relevant but simplified model. This model consists of two layers of quantum Hall ferromagnet defined in opposite magnetic fields, coupled by an anti-ferromagnetic interaction $J_{x,y,z}$. As mentioned above, the layer ($\gamma^z$) and spin degrees of freedom in the model mimic the Chern sector and valley pseudo-spin in the original problem of TBG. The effective Hamiltonian of this LLL model is:
\begin{eqnarray}
H &=& \psi^{\dag}\frac{(\bm p+e\gamma^z\bm A)}{2m}\psi + \frac{1}{2}\int :n(\bm r)V_C(\bm r-\bm r')n(\bm r'): \nonumber\\
&& - \frac{e^2l_B}{4\pi\epsilon}\sum_{i=x,y,z}J_i:\left(\psi^{\dag}\gamma^z\eta^i\psi(\bm r)\right)^2:,\label{skyrmionmodel}
\end{eqnarray}
where $V_C$ is the Coulomb interaction, $n(\bm r) = \sum_{\gamma\eta}\psi^{\dag}_{\gamma\eta}\psi_{\gamma\eta}(\bm r)$ is the charge density, and $J_x,J_y,J_z$ are the anti-ferromagnetic XXZ interactions between the two layers. The cyclotron gap between Landau levels is much larger than other interacting scales, so that only the LLL degrees of freedom are physically relevant. Ref.~(\onlinecite{skyrmionSC_DMRG}) showed numerical evidence for superconductivity and computed the skyrmion binding energy. However, one crucial difference between this LLL model and realistic TBG is the absence of the band geometry inhomogeneity in Eqn.~(\ref{skyrmionmodel}). Going beyond the LLL limit by allowing nonuniform band geometry is fundamentally important and interesting.

There are a couple of places that band geometry could enter into the model Eqn.~(\ref{skyrmionmodel}). First, band geometry modifies the stiffness of spin waves and thus determines the details of a single skyrmion, such as its size and elastic energy. Second, band geometry modifies the electron-electron interaction through form factors and therefore changes the effective skyrmion-skyrmion interactions. In this section, we pay particular attention to how band geometry would modify the electron-electron interaction by using the ideal flatband theory derived in the main text. For this purpose, we approximate the flatband wavefunction as their chiral limit wavefunctions, and assume finite skyrmion coupling even in the exact flat limit. At magic angle, the wavefunctions for the four flatbands are:
\begin{eqnarray}
\Psi_{\bm k}^{+,+1} &=& \left(\begin{matrix}\psi_{\bm k}(\bm r)\\0\\0\\0\end{matrix}\right),\quad \Psi_{\bm k}^{+,-1} = \left(\begin{matrix}0\\\psi^*_{\bm k}(-\bm r)\\0\\0\end{matrix}\right),\nonumber\\
\Psi_{\bm k}^{-,-1} &=& \left(\begin{matrix}0\\0\\\psi^*_{-\bm k}(\bm r)\\0\end{matrix}\right),\quad \Psi_{\bm k}^{-,+1} = \left(\begin{matrix}0\\0\\0\\\psi^*_{-\bm k}(-\bm r)\end{matrix}\right),\nonumber
\end{eqnarray}
where each wavefunction has four components labeled by $(+A,+B,-A,-B)$ which represent the valley $\pm$ and the sublattice $A/B$ degrees of freedom. Due to the $C_2$ and $T$ symmetry, the wavefunctions are constrained to take the above form and consequently the form factors are diagonal in valley and sublattice indexes \cite{Zaletel_PRX20}. The $\psi_{\bm k}(\bm r)$ is a two-component layer spinor given in Eqn.~(\ref{cTBGwf}).

Following the discussion in the main text, the magic angle chiral TBG wavefunctions $\Psi^{\alpha}_{\bm k}$ above allow us to derive an effective COM dependent interaction in the LLL basis. The interacting Hamiltonian, written in the LLL basis, is:
\begin{eqnarray}
H &=& \tilde{v}^{\alpha\beta}_{\bm k_1\bm k_2\bm k_3\bm k_4}c_{\bm k_1}^{\alpha\dag}c_{\bm k_2}^{\beta\dag}c_{\bm k_3}^{\beta}c_{\bm k_4}^{\alpha},
\end{eqnarray}
where $c^{\alpha\dag}_{\bm k}$ creates a LLL electron in a magnetic field where the sign of the magnetic field depends on sublattice. More precisely, its wavefunction $\Phi^{\alpha}_{\bm k}(\bm r)=\langle\bm r|c^{\alpha\dag}_{\bm k}|0\rangle$ is:
\begin{equation}
\Phi^{(v,c=+1)}_{\bm k} = \Phi_{\bm k}(\bm r),\quad\Phi^{(v,c=-1)}_{\bm k} = \Phi^*_{\bm k}(-\bm r),
\end{equation}
where $\Phi_{\bm k}(\bm r)$ was defined and discussed extensively around Eqn.~(\ref{qhwf}); $\Phi^*_{\bm k}(-\bm r)$ is its time-reversal conjugate. Note that the basis $c^{\alpha}_{\bm k}$ is precisely the basis used in the model Eqn.~(\ref{skyrmionmodel}). If using Coulomb interaction and assuming skyrmion coupling $J_{x,y,z}$ we recover the model Eqn.~(\ref{skyrmionmodel}).

However, in the presence of inhomogeneous band geometries, the interaction matrix element $\tilde{v}^{\alpha\beta}_{\bm k_1\bm k_2\bm k_3\bm k_4}$ is no longer purely Coulomb, but modified by an effective COM interaction. After some algebra, it can be shown that:
\begin{eqnarray}
\tilde{v}^{\alpha\beta}_{\bm k_1\bm k_2\bm k_3\bm k_4}\! =\! \int_{\bm r_{1,2}}\!\tilde{V}_C(\bm r_1,\bm r_2)\Phi^{\alpha*}_{\bm k_1}\Phi^{\alpha}_{\bm k_4}(\bm r_1)\Phi^{\beta*}_{\bm k_2}\Phi^{\beta}_{\bm k_3}(\bm r_2),\nonumber
\end{eqnarray}
where the effective interaction is given as follows, which is the key result of this section:
\begin{equation}
\tilde{V}_C(\bm r_1,\bm r_2) = |\mathcal{B}(\bm r_1)|^2|\mathcal{B}(\bm r_2)|^2V_C(\bm r_1-\bm r_2),\label{effectiveVC}
\end{equation}
with $|\mathcal{B}(\bm r)|^2 = |\mathcal{G}(\bm r)|^2+|\mathcal{G}(-\bm r)|^2$ following Eqn.~(\ref{cTBGwf}). Eqn.~(\ref{effectiveVC}) shows that the band geometry $\mathcal{B}(\bm r)$ modifies the Coulomb interaction $V_C$ in the LLL basis by making it COM dependent. Moreover, the modified effective interaction is essentially in the same form as the single band case. This indicates that the COM interaction generally has attractive components as seen from its COM pseudopotential decomposition shown in the main text in FIG. 3(b). In the single band problem, such attractive interaction drives the FQH to CDW transition, and we anticipate in the four band problem here such interaction could enhance the bonding of skyrmions to enhance superconductivity, suppose the spatial pattern of the COM interaction matches that of the order parameter of superconductivity.

Moreover, regarding $|\mathcal{B}(\bm r)|^2$ as the Umklapp processes that scatters electron across the Brillouin zone, we know that the shortest distance scattering terms dominate, which practically allows us to incorporate the effect of band geometry into the skyrmion model Eqn.~(\ref{skyrmionmodel}) by just adding \emph{one} more parameter following similar discussions around Eqn.~(\ref{simplemodel}). This could simplify the calculation of interacting physics in inhomogeneous band geometry backgrounds. A thorough understanding of the skyrmion physics in such systems requires extensive numerical studies which we leave for future work. We believe the theoretical tools developed in this work makes such large-scale numerical analysis possible.

\subsection{Stability of the composite Fermi liquid phase in TBG flatbands}
The theoretical tool developed in this work has wide applications in studying the stability of various many body phases against inhomogeneous band geometries. As the second example of applications, we study the stability of composite Fermi liquid in TBG flatband.

The composite Fermi liquids (CFL) are gapless phases that occur at even denominator filling fractions in the LLL $\nu=1/(2m)$ where $m$ is an integer. They are understood as Fermi seas of composite fermions, where each composite fermion consists of one electron and $2m$ flux quanta. Unlike FQH states such as Laughlin or Moore-Read states which have been realized in flatband models, there are rare reports of CFL in lattice models. In this section, we study the stability of CFL against nonuniform band geometries. Particularly we ask whether CFL can exist in TBG topological flatbands.

Studying CFL in flatband models is harder than that in the continuum LLLs, due to the lack of continuous translation symmetry in lattice systems. In the LLL, continuous translation symmetry implies eigenstates are distinguished by $N_e\times N_{\phi}$ momentum quantum numbers where $N_{\phi}$ is the total flux quanta and $N_e$ is the particle number. Moreover, each eigenstate can be compared by wavefunction overlap with one many-body model wavefunction following Ref.~(\onlinecite{scottjiehaldane}) and Ref.~(\onlinecite{Jie_Dirac}). In flatbands, however, lattice rather than continuous translation symmetry mean that there are only $N_{\phi}$ good translation quantum numbers. This means exact diagonalization states are down-folded into a smaller Brillouin zone, and thereby need to be compared with a linear combination of a couple of model wavefunctions which practically makes the wavefunction overlap approach much harder.

While directly comparing the wavefunction is harder, we can use spectral information to identify many-body phases. Following the main text, the interaction in an ideal topological flatband is mapped into the LLL basis as shown in Eqn.~(\ref{FQHumklapp}). Equivalently, the interacting physics is described by a FQH problem with the COM interaction shown in Eqn.~(\ref{simplemodel}). This allows to \emph{continuously} tune the COM parameter $w_{\bm b}$ to track the evolution of the spectrum and see explicitly the influence from band geometry.

To model the inhomogeneous band geometries, we set $w_{\bm 0}=1$ and $w_{\bm b_1}\neq0$, which we have shown is a very good approximation to the real TBG spectrum. Then $w_{\bm b_1}$ is the only tuning parameter in our problem: when $w_{\bm b_1}=0$ we recover the LLL limit, and when $w_{\bm b_1}=0.243$ we recover the flatband of cTBG at the first magic angle. See Eqn.~(\ref{u0u1value}) for calculation details.

\begin{figure*}[h]
    \centering
    \includegraphics[width=1\textwidth]{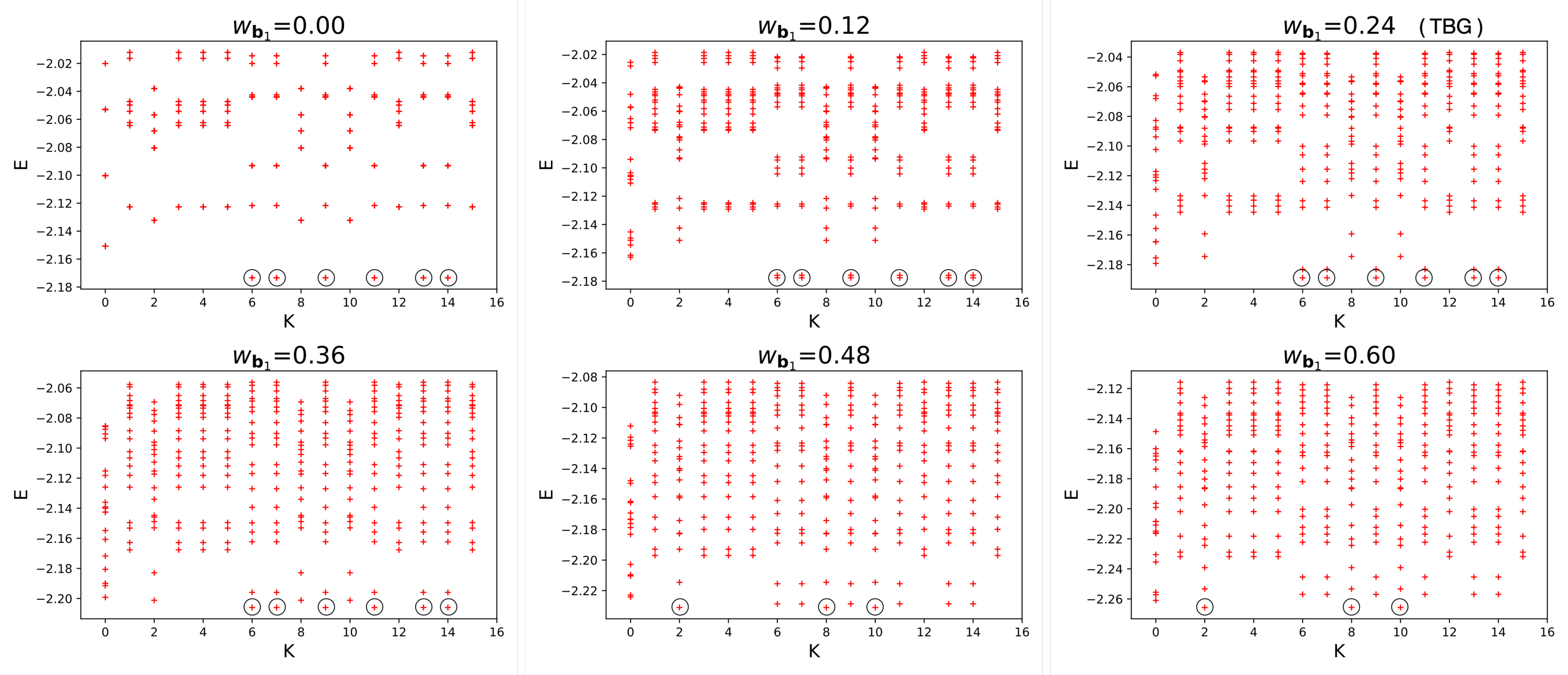}
    \caption{Energy spectrum of a half filled LLL subject to an interaction consisting of the Coulomb interaction plus a center-of-mass interaction parameterized by $w_{\bm b_1}$; the center-of-mass interaction represents the effect of an inhomogeneous band geometry. The LLL limit is recovered when $w_{\bm b_1}=0$, and a cTBG flatband is recovered when $w_{\bm b_1}=0.243$. Each figure is plotted with respect to the lattice translation momentum $K=k_x+N_x\times k_y$ where $k_{x,y}$ are two momentum quantum numbers of two spatial directions.}\label{sup_pic1}
\end{figure*}

The exact diagonalization spectrum for $N_e=8$ electrons on an $N_x\times N_y=4\times4$ TBG superlattice is shown in FIG.~\ref{sup_pic1}. We plot the spectrum in an ascending order of $w_{\bm b_1}$ from $0$ to $0.6$. Energies in each figure are plotted with respect to the quantum number $K=k_x+N_x\times k_y$, where $k_x\in[0,N_x)$ and $k_y\in[0,N_y)$ are translation quantum numbers of two spatial directions. Empty circles mark the ground states in each figure.

We notice that in the LLL limit ($w_{\bm b_1}=0$), ground states occur at $K=6,7,9,13,14$, which in fact are all double degenerate (due to COM topological degeneracy on torus). These quantum numbers agree with the dipole momentum of CFL: the most compact Fermi sea configuration for $N_e=8$ particles on $4\times4$ hexagonal lattice is shown in the FIG.~\ref{sup_pic2} a) as the seven dots plus one of the triangles. These Fermi sea momenta are precisely $6,7,9,13,14$, reflecting the fact that the Coulomb ground state at half filled LLL is CFL. Upon turning on a finite but small $w_{\bm b_1}$ (representing a small inhomogeneous band geometry background), we found the two-fold topological degeneracy is split and the spectrum is slightly modified. This is consistent with the fact that inhomogeneous band geometry breaks the COM translation symmetry and lifts the degeneracy. When further increasing the $w_{\bm b_1}$, we found energy levels of other momentum sectors $2,8,10$ gradually become lower. When $w_{\bm b_1}$ is in between $0.36$ and $0.48$, a first order phase transition occurs, after which the ground states exist at $K=2,8,10$ which does not correspond to the lowest energy configuration of a compact Fermi sea.

\begin{figure}[h]
    \centering
    \includegraphics[width=0.45\textwidth]{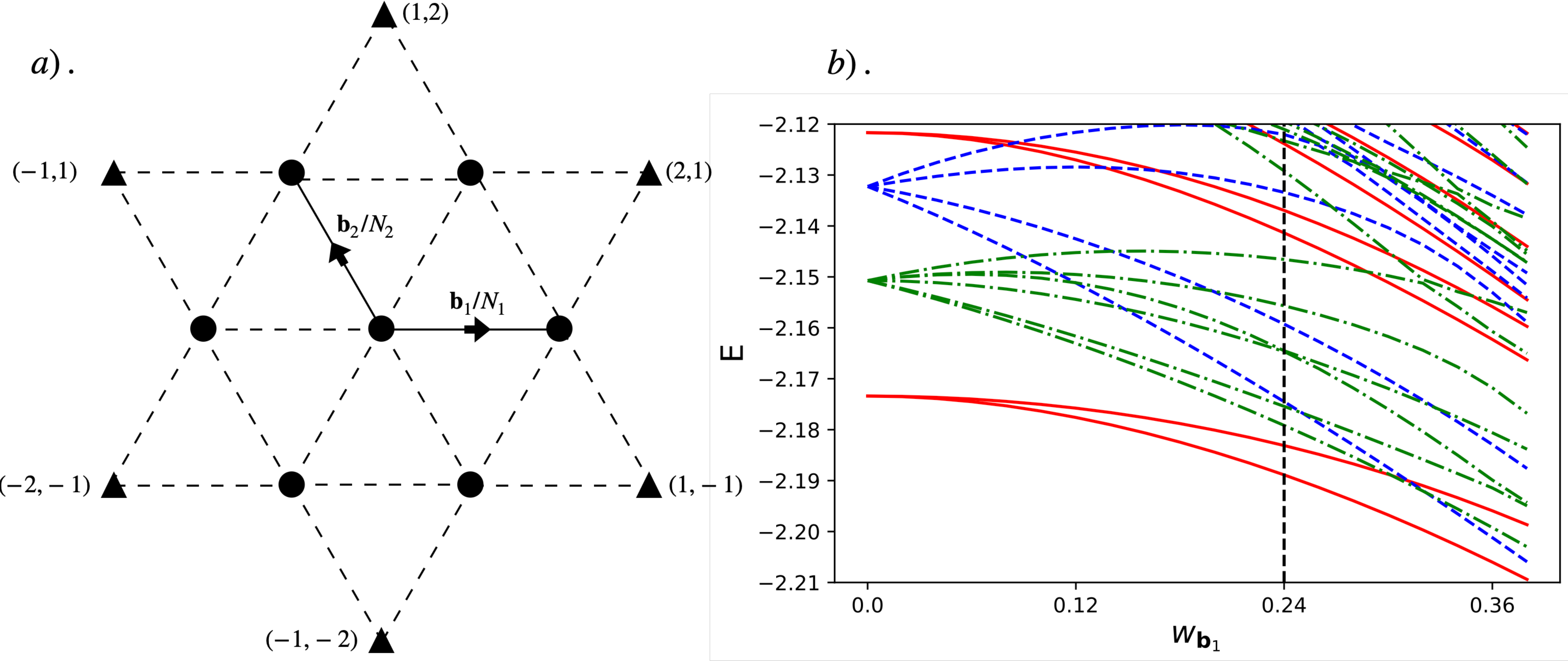}
    \caption{a). The seven dots plus one of the six triangles represent the configuration of composite fermion dipoles of lowest energies for $N=8$ particles on a $N_x\times N_y=4\times4$ TBG superlattice (hexagonal torus). These configurations agree with the lowest momentum sectors of the Coulomb interaction at half filling for both the LLL and cTBG flatband. b) Evolution of the energy spectrum as a function of inhomogeneous band geometry represented by $w_{\bm b_1}$ at fixed momentum ($K=6,0,2$ for the red solid, green dashed-dotted, and blue dashed line, respectively). The cTBG has $w_{\bm b_1}=0.24$ as marked by the dashed line. This result shows that the ground states of cTBG are qualitatively identical to those of the LLL as there is no phase transition when continuously tuning the band geometry. It also shows that the excitation gaps to high energy non-universal states are reduced by a varying band geometry.}\label{sup_pic2}
\end{figure}

The cTBG flatband has the value $w_{\bm b_1}=0.24$. The evolution of the spectrum while continuously tuning band geometry is shown in FIG.~\ref{sup_pic1} b) and FIG.~\ref{sup_pic2}, which show that there is no level crossing at $w_{\bm b_1}<0.243$. This indicates the Coulomb ground states of cTBG are qualitatively identical to those in the LLL. However, from the diagonalization result we also see that the excitation gaps to high energy non-universal states are reduced significantly by band geometry inhomogeneous. Thereby we conclude that a CFL exists in chiral TBG flatbands, but is less stable than in the LLL. We leave more detailed studies on larger system sizes, wavefunction overlap, as well as questions beyond the chiral limit, to future work.

% \bibliography{ref.bib}

\end{document}